\documentclass[article,11pt,authoryear]{elsarticle}
\usepackage{mydef}
\usepackage{amstext}
\usepackage{afterpage}
\usepackage[hide]{ed}
\def\rhod{\widetilde{\rho}}

\journal{
}
\begin{document}

\begin{frontmatter}

\title{Modeling stochastic skew of FX options using SLV models with stochastic spot/vol correlation and correlated jumps
}

\author[ny]{A.~Itkin
}
\ead{aitkin@nyu.edu}
\cortext[cor1]{Corresponding author}
\address[ny]{Tendon School of Engineering, New York University, \\
12 Metro Tech Center, RH 517E, Brooklyn NY 11201, USA}


\begin{abstract}
It is known that the implied volatility skew of FX options demonstrates a stochastic behavior which is called {\it stochastic skew}. In this paper we create stochastic skew by assuming the spot/instantaneous variance correlation to be stochastic. Accordingly, we consider a class of SLV models with stochastic correlation where all drivers - the spot, instantaneous variance and their correlation are modeled by \LY processes. We assume all diffusion components to be fully correlated as well as all jump components. A new fully implicit splitting finite-difference scheme is proposed for solving forward PIDE which is used when calibrating the model to market prices of the FX options with different strikes and maturities. The scheme is unconditionally stable, of second order of approximation in time and space, and achieves a linear complexity in each spatial direction. The results of simulation obtained by using this model demonstrate capacity of the presented approach in modeling stochastic skew.
\end{abstract}

\begin{keyword}
stochastic correlation \sep FX options \sep stochastic skew \sep SLV models \sep correlated jumps \sep 3D PIDE \sep finite-difference \sep forward equations \sep fully implicit splitting scheme, unconditional stability

\JEL C6, C61, G17
\end{keyword}

\end{frontmatter}

\section{Introduction}

It is known that in the FX world skew of the implied volatility of the FX options demonstrates some kind of stochastic behavior. For instance, \cite{CarrWu2007} report that for two currency pairs (which are the U.S. dollar prices of the Japanese yen and the British pound) by analyzing the market implied volatilities of option series they find the slope of the smile to greatly vary over time. Namely, the sign of the slope switches several times in the considered sample. Therefore, although the risk-neutral distribution of the
currency return exhibits persistent fat-tail behavior, the risk-neutral skewness of the distribution experiences strong time-variation which can be very positive or very negative on any given date.

To capture this effect when doing option pricing, new models are required because, as mentioned in \cite{CarrWu2007},  the jump-diffusion stochastic volatility models like that in  \cite{Bates:96}, are able to capture the average shape of implied volatility smiles and the time-variation in their levels but  cannot generate strong time variation in the risk-neutral skewness of currency returns. We can mention at least two approaches that are used to attack this problem.

The first one is proposed in \cite{cw04} where skew is modeled by separate the up jumps and the down jumps in spot. These jumps are modeled by two \LY processes. The stochastic volatility and skewness are introduced by a separate time change in each \LY component.  The stochastic variation in the relative proportion of up and down jumps generates stochastic variation in the risk-neutral skewness of currency returns. The model is analytically tractable for pricing European vanilla options. However, as per \cite{Higgins2014}
it often requires jumps in spot that are much larger than those actually experienced. Therefore, it can misprice barrier options if the probability of spot jumping through the barrier is significant Also pricing of exotic options requires numerical methods, see e.g., \cite{ItkinCarrBarrierR3}.

Another approach is to randomize those model parameters that govern the risk-neutral skewness. In jump-diffusion models they are the mean jump size and the correlation coefficient between the currency return and the stochastic variance process. The first parameter governs the risk-neutral skewness at short maturities, and the second one - at long maturities. Therefore, in a set of papers the correlation between the spot and instantaneous variance is considered to be stochastic and driven by a separate stochastic process, see \cite{Emmerich2006,FGT2007,Ma2009,Ahdida2013,Higgins2014,Zetocha2015} and references therein. Moreover, in \cite{Zetocha2015b} where a pricing problem is considered for basket equity options, the standard correlation model, which is based on the Jacobi process, is extended by adding jumps in the correlation process.  This is done based on the analysis of the market data for basket options which demonstrates that the jump to high level of correlation occurs after several down moves that drag the basket below some comfort zone of normality. The jumps are modeled in a simplified way by adding to the Jacobi process a term proportional to $dN$ with $N$ being a Poisson process whose intensity depends on the basket performance and time.

\section{Modeling stochastic correlation}

We mention two basic approaches to model the stochastic correlation $\rho_t$. Consider first only the diffusion processes. Since by definition $\rho_t$ is bounded in $[-1,1]$, the first approach would be to model it as a bounded diffusion. One of the popular choices is the Jacobi process used, e.g., in \cite{Emmerich2006, Zetocha2015}. However, this requires the process to be mean-reverting, while it is not obvious why the correlation should be mean-reverted.\footnote{Perhaps, the only reason would be if one wants an expectation of the stochastic correlation to be stationary at long time horizons. However, we don't know any paper that would justified this based on the available market data.} Other bounded diffusions can also be used.

This approach is tractable for European vanilla options, but not for the exotic options. Because of that, in \cite{Higgins2014}  an analytical approximation for barrier and one-touch options is developed which is based on semi-static vega replication. It is shown that this approximation works well in markets where the risk neutral drift is modest.

The second approach, \cite{Emmerich2006}, uses a transformation $\rho_t = g(X_t)$ where $X_t$ is an arbitrary diffusion process (perhaps driftless), and function $g(X_t)$ transforms the support of the $X_t$ outcomes to $[-1,1]$. Popular transformations are hyperbolic tangent, \cite{Teng2016}, normal CDF, \cite{Carr2012}, normalized inverse tangent, \cite{Emmerich2006}, etc. In principle, any continuous map $\mathbb{R} \to [-1,1]$ can be used for this purpose, so the real choice becomes a matter of having some additional properties, e.g., tractability, to be taken into account. Also, despite this approach is a bit less intuitive, it opens the door to various sophisticated models of the stochastic correlations.

One of such possible extensions would be a natural idea to consider the underlying process $X_t$ to be a jump-diffusion or a \LY process. To the best of our knowledge, there is the only attempt to model correlation by introducing jumps which is presented in \cite{Zetocha2015}. However, this model is very simplistic and empirical. It is introduced based on the analysis of some market data on baskets of stocks, with an observation that the jump to high level of correlation might happen after several down moves that put the basket below some level of normality. Therefore, when modeling these baskets using the stochastic correlation model, \cite{Zetocha2015} defines $\rho_t$ to follow the extended Jacobi process
\begin{equation} \label{jac}
d \rho_t = \alpha(\overline{\rho} - \rho_t) + \beta \sqrt{(1-\rho_t)(\rho_t - \rho_{min})}d W_t + (1-\rho_t)d N_t \cdot \mathbf{1}_{I_t<I_0}.
\end{equation}
Hear $\alpha, \beta, \rho_{min}$ are some constant model parameters, $W_t$ is the standard Brownian motion, $N_t$ is the Poisson jump process with the intensity $\lambda$ given by the following expression
\begin{equation}
\lambda = \lambda_0 \left(\dfrac{1-I_t}{1-I_c(t)}\right)^\kappa,
\end{equation}
\noindent where $\lambda_0, \kappa$ are constants, $I_c(t)$ is the critical level defined as $n$ standard deviations of the basket performance $I_t$
\begin{equation}
I_c(t) = F(t) e^{-\frac{1}{2}\sigma^2 t - n \sigma \sqrt{t}},
\end{equation}
Here $\sigma$ is the at-the-money volatility of the basket at time $t$, and $F(t)$ is the basket forward at time $t$. So the jumps can occur if the basket performance $I_t$ hits the level $I_t = I_0$ at $t > 0$. Otherwise, if $I_t - I_0$ is positive the jump doesn't occur. The jump size $(1-\rho_t)$ in \eqref{jac} is chosen in such a way, that if the jump occurs, $\rho_t$ always jumps from $\rho_{t^-}$ to 1 (full correlation).

Obviously, this model is empirical with a lot of parameters introduced, which financial meaning is not entirely transparent. Therefore, calibration of such a model could be a hard issue. Also the assumption that the jump always results in the full correlation to be reached immediately after the jump occurs, seems to not be justified at all.

Therefore, in this paper we utilize a different approach. We use a general \LY models framework to introduce jumps, see e.g., \cite{ContTankov}. Since a general \LY process is not bounded, we use a certain transformation as this is described in above, to bound it to $[-1,1]$. Once this is done, we consider three stochastic processes: the spot, instantaneous variance and their correlation each driven by some \LY process. We assume that the Brownian components of each process are correlated, and the jump component are correlated similar to how this is done in \cite{ItkinLipton2014,Itkin3D}. However, the Brownian motions and the jump components remain uncorrelated. This model has a clear financial meaning, and can be calibrated in few steps. For instance, idiosyncratic components of jumps can be calibrated separately to the corresponding market; the local volatility functions can be calibrated to a set of FX European vanilla options; then stochastic volatility and correlation parameters (the diffusion part) of the model can be calibrated to FX exotic options which demonstrate a smile. Finally, the common jumps parameters of the model can be calibrated to the implied volatility skew of the exotic options. This procedure runs in a loop until it converges.

The main contribution of the paper consists in several results. First, a SLV model for pricing exotic FX options which demonstrate the existence of stochastic skew is proposed which assumes the spot/InV correlation to be stochastic. Moreover, the model also includes correlated \LY jumps in each of the stochastic drivers. Second, a new fully implicit finite-difference splitting scheme is proposed to solve a forward Kolmogorov equation which is of second order in both temporal and each spatial dimension, is unconditionally stable and preserves both positivity and norm of the solution while complexity-wise is linear in the number of grid nodes in each spatial dimension. We also formulate and prove a theorem that the proposed discretization preserves the norm of the solution.

The rest of the paper is organized as follows. In Section~\ref{sect:Model} we formulate the model which considers stochastic correlation of Brownian drivers and deterministic correlation of \LY jump processes. Section~\ref{sectPIDE} discusses the corresponding backward PIDE and boundary and initial conditions. In Section~\ref{sect:Forward} same consideration is given to the Forward Kolmogorov equation. In Section~\ref{sect:FD} we introduce a new fully-implicit finite-difference scheme for this forward PIDE and analyze its properties. Results of some numerical  experiments obtained by using this scheme are provided in Section~\ref{sect:numerics}. An Appendix provides a detailed derivation of the forward FD scheme.

\section{Model \label{sect:Model}}

We consider an LSV model with stochastic spot/InV correlation and jumps by introducing stochastic dynamics for variables $S_t,v_t,R_t$. Here $S_t$ is the spot price, $v_t$ is the instantaneous variance, and
\begin{equation}
R_t = \arctanh(\rho_t).
\end{equation}
In other words, the stochastic correlation $\rho_t$ is represented as $\rho_t = \tanh (R_t)$. Therefore, since $R_t \in \mathbb{R}$, this map guarantees that $\rho_t \in [-1,1]$.

Similar to \cite{Itkin3D}, we consider a model where the stochastic SDE for each variable $S_t,v_t,R_t$ includes both diffusion and jumps components, as follows:
\begin{align} \label{sde}
d S_t &= (r_d - r_f)S_t dt + \Sigma_{s}(S_t,t) S_t^c\sqrt{v_t} W_{s,t} + S_t d L_{S,t}, \\
d v_t &= \kappa_v(t)[\theta_v(t) - v_t] dt + \xi_v v_t^a W_{v,t} + v_t dL_{v,t}, \nonumber \\
d R_t &= \kappa_r(t)(\theta_r(t) - R_t) dt + \xi_r W_{r,t} + dL_{r,t}, \nonumber
\end{align}
\noindent subject to the following initial conditions
\begin{equation}
S_t\Big|_{t=0} = S_0, \quad v_t\Big|_{t=0} = v_0, \quad R_t\Big|_{t=0} = R_0 = \arctanh(\rho_0).
\end{equation}
Here $r_d, r_f$ are the domestic and foreign interest rates, $t$ is the time, $\Sigma_s$ is the local volatility function, $W_{s,t}, W_{v,t}, W_{r,t}$ are the correlated Brownian motions, $\kappa_v, \theta_v, \xi_v$ are the mean-reversion rate, mean-reversion level and volatility of volatility (vol-of-vol) for the instantaneous variance $v_t, \kappa_r, \theta_r, \xi_r$ are the corresponding parameters for $R_t$, $0 \le a < 2, \ 0 \le c < 2$ are some constants (power constants) which are introduced to add an additional flexibility to the model as compared with the popular Heston ($\alpha = 0.5$), lognormal ($\alpha = 1$) and 3/2 ($\alpha = 1.5$) models. As mentioned in \cite{Itkin3D}, one has to be careful if she wants to determine these parameters by calibration. This is because having both the vol-of-vol and the power constant in the same diffusion term brings an ambiguity into the calibration procedure. However, it can be resolved if some additional financial instruments are used in calibration, e.g., exotic option prices are combined with the variance swaps prices, see \cite{Itkin2013}.

In the last equation of \eqref{sde} the drift term is introduced to be mean-reverting. However, this might not be necessary\footnote{Or, as it was already mentioned, there is no evidence that the stochastic correlation should be mean-reverting.}. Therefore, when doing so we rely on calibration of the model to the market data. In other words, if the data exhibits mean-reversion, we expect the mean-reversion level $\theta_r(t)$, found by the calibration, to differ from zero. On contrary, if $\theta_r(t) = 0$, there is no mean-reversion, and the sign of the drift term depends on the sign of $\kappa_r(t)$.

Processes $L_{s,t}, L_{v,t}, L_{r,t}$ are pure discontinuous jump processes with generator $A$
\[
A f(x) = \int_{\mathbb{R}}\left( f(x+y) - f(x) - y \mathbf{1}_{|y|<1}\right) \mu (dy),
\]
\noindent with $\mu (dy)$ be a L{\'e}vy measure, and
\[
\int_{|y|>1}e^{y}\mu (dy)<\infty.
\]
At this stage, the jump measure $\mu (dx)$ is left unspecified, so all types of jumps including those with finite and infinite variation, and finite and infinite activity could be considered here.

The last line of \eqref{sde} is a combination of the Hull-White model, \cite{hull/white:90}, for the diffusion component with a (arithmetic) Levy process for the jump component. As such, this model allows both positive and negative values of $R_t$.

\subsection{Correlations}

Below for a better transparency we consider three cases.

\paragraph{\bf No jumps} Consider first the case when there is no jumps in \eqref{sde}. Then
\begin{align} \label{corrDef}
\left<dW_{s,t} dW_{r,t}\right> &= \rhod_{sr} dt, \\
\left<dW_{v,t} dW_{r,t}\right> &= \rhod_{vr} dt, \nonumber \\
\left<dW_{s,t} dW_{v,t}\right> &= \rhod_{t} dt = \tanh(\widetilde{R}_t), \nonumber
\end{align}
\noindent where we use the symbol $\widetilde{}$ to mark parameters and variables for a pure diffusion case. The last line in \eqref{corrDef} is explained in detail in \cite{Emmerich2006}. It means that
\[ \E[W_{s,t} W_{r,t}] = \E \left[\int_0^t \rhod_s ds\right] \ne \rhod_t t \]
\noindent as it would be if $\rhod_t$ is constant.

We assume that $\rhod_{vr}$ is constant, but not $\rhod_{sr}$. Indeed, the correlation matrix
\begin{equation}
\mathcal{R} =
\begin{pmatrix}
1 & \rhod_t & \rhod_{sr} \\
\rhod_t & 1 & \rhod_{vr} \\
\rhod_{sr} & \rhod_{vr} & 1 \\
\end{pmatrix}
\end{equation}
\noindent must be positive-semidefinite. This provides restrictions on, at least, one of the correlation coefficients $\rhod_{sr}, \rhod_{vr}$. For instance, let us assume that
$\rhod_{sr} \in [-1,1]$ and then, since also $\rhod_t \in [-1,1]$, the semi-definiteness is defined by the condition
\begin{equation} \label{posSemiDef}
det(\mathcal{R}) = 1 - \rhod_t^2 - \rhod_{sr}^2 - \rhod_{vr}^2 + 2 \rhod_t \rhod_{vr} \rhod_{sr} \ge 0,
\end{equation}
\noindent which is solved by
\[ -B - \rhod_t \rhod_{sr} \le \rhod_{vr} \le B - \rhod_t \rhod_{sr}, \qquad B = \sqrt{(1-\rhod_t^2)(1-\rhod_{sr}^2)}.\]
These allowed values of $\rhod_{vr}$ are depicted in Fig.~\ref{Fig1}.
\begin{figure}[!ht]
\begin{center}
\fbox{\includegraphics[width=4 in]{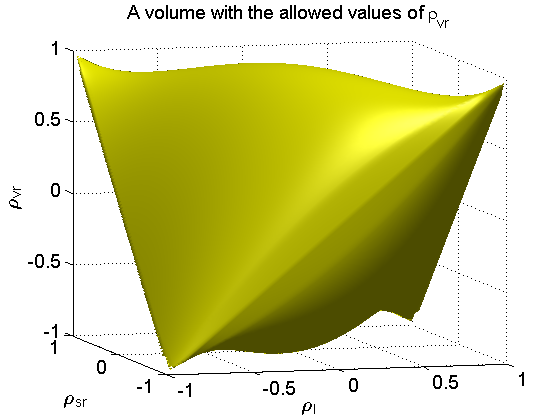}}
\caption{A volume with the allowed values of $\rhod_{vr}$ (inside the shaded area) as a function of $\rhod_t, \rhod_{sr}$.}
\label{Fig1}
\end{center}
\end{figure}
So the domain of definition of $\rhod_{vr}$ depends on $\rhod_t$, and since $\rhod_t$ is stochastic, it also changes stochastically. Therefore, $\rhod_{vr}$ cannot be a constant.

To resolve this issue, following the idea of \cite{Higgins2014}, further we assume that
$\rhod_{vr} = const, \ \rhod_{sr} = \rhod_t \rhod_{vr}$. Thus, the spot/correlation correlation is also stochastic, but driven by the $\rhod_t$ process multiplied by a constant. This provides
\[  det(\mathcal{R}) = (1 - \rhod_t^2)(1-\rhod_{vr}^2) \ge 0,  \qquad \forall \ (\rhod_t \in [-1,1]) \cup (\rhod_{vr} \in [-1,1]). \]

\paragraph{\bf Pure jump processes} It is known that models with jumps in the spot price better predict the observed market data, especially at short maturities. It is also known, that jumps in $v_t$ might also be needed. For instance, in \cite{Sepp2011a,Sepp2011b} exponential and discrete jumps are investigated in both $S_t$ and $v_t$. The author concludes that infrequent negative jumps in $v_t$ are necessary to fit the market data on equity options. Definitely, jumps in both spot and the instantaneous variance could be correlated. Then it would be natural to extend this point and introduce the correlation between the jumps in $\rho_t$ and that in $S_t$, as well as between jumps in $\rho_t$ and that in $v_t$. For instance, \cite{ClementsLiao2013} considered  the links
between co-jumps within a group of large stocks, the volatility of, and correlation between
their returns. They found that, despite the occurrence of common, or co-jumps between the stocks
is unrelated to the level of volatility or correlation, both volatility and correlation are lower subsequent to a co-jump. Therefore, it does make sense to consider jumps in the spot, instantaneous variance and their correlation to be fully correlated.

In this paper we introduce these correlations as this is done in \cite{Itkin3D, ItkinLipton2014} by following the approach of \cite{Ballotta2014}. They construct the jump process as a linear combination of two independent L{\'e}vy processes representing the systematic factor and the idiosyncratic shock, respectively (see also \cite{ContTankov}). It has an intuitive economic interpretation and retains nice tractability, as the multivariate characteristic function in this model is available in closed form based on the following proposition of \cite{Ballotta2014}:
\begin{proposition} \label{propBallotta}
Let $Z_t, \ Y_{j,t}, \ j=1,...,n$ be independent L{\'e}vy
processes on a probability space $(Q, F, P )$, with characteristic functions
$\phi_Z(u; t)$ and $\phi_{Y_j}(u; t)$, for $j=1,...,n$ respectively. Then,
for $b_j \in \mathbb{R}, \ j=1,...,n$
\[
X_t = (X_{1,t},..., X_{n,t})^\top = (Y_{1,t} + b_1 Z_t,...,Y_{n,t} + b_n
Z_t)^\top
\]
is a L{\'e}vy process on $\mathbb{R}^n$. The resulting characteristic function
is
\[
\phi_{\mathbf{X}}(\mathbf{u}; t) = \phi_Z\left( \sum_{i=1}^n b_i u_i;
t\right) \prod_{i=1}^n \phi_{Y_j}(u_j;t), \qquad \mathbf{u} \in \mathbb{R}^n.
\]
\end{proposition}

By construction every factor $X_{i,t}, i=1,...,n$ includes a common factor $Z_t$. Therefore, all components $X_{i,t}, i=1,...,n$ could jump together, and loading factors $b_i$ determine the magnitude (intensity) of the jump in $X_{i,t}$ due to the jump in $Z_t$. Thus, all components of the multivariate
L{\'e}vy process $\mathbf{X}_t$ are dependent, and their pairwise correlation is given by (again
see \cite{Ballotta2014} and references therein)
\begin{equation} \label{jumpCorr}
\rho_{ij} = \dfrac{b_j b_i \mbox{Var}(Z_1)} {\sqrt{\Var(X_{j,1})}\sqrt{\Var(X_{i,1})}}.
\end{equation}
Such a construction has multiple advantages, namely:

\begin{enumerate}
\item As $\mbox{sign}(\rho_{i,j}) = \mbox{sign}(b_i b_j)$, both positive and
negative correlations can be accommodated

\item In the limiting case $b_i \to 0$ or $b_j \to 0$ or $\Var(Z_1) =
0 $ the margins become independent, and $\rho_{i,j} = 0$. The other limit $b_i \to \infty$ or $b_j \to \infty$ represents a full positive correlation
case, so $\rho_{i,j} = 1$. Accordingly, $b_i \to \infty, \ b_{3-i} \to
\infty, \ i=1,2$ represents a full negative correlation case as in this
limit $\rho_{i,j} = -1$.
\end{enumerate}
If we want the jump correlations to be constant, then this setting is sufficient. However, since we want the correlation $\rho_{sv} \equiv \rho_t$ to be stochastic, this approach has to be modified. In this paper we don't elaborate on it, and use the former one, so assuming that the correlation between jumps is constant.

\paragraph{\bf General case} When the underlying \LY processes have both the diffusion and jumps components, the correlation matrix can be build as a combination of two constructions presented in above. We assume that all Brownian motions $W_{i,t}, \ i \in [S,v,r]$ are independent of $L_{i,t}$. In this setting the total instantaneous correlations read
\begin{align}  \label{corr}
\rho_{sv} &= \dfrac{\rhod_t \sigma_s \sigma_v + b_s b_v \Var(Z_1)}{\sqrt{\sigma_s^2 + S^2 \Var(L_{s,1})} \sqrt{\sigma_v^2 + v^2 \Var(L_{v,1})}}, \\
\rho_{vr} &= \dfrac{\rhod_{vr} \sigma_v \sigma_r + b_v b_r \Var(Z_1)}{\sqrt{\sigma_v^2 + v^2 \Var(L_{v,1})} \sqrt{\sigma_r^2 + \Var(L_{r,1})}}, \nonumber  \\
\rho_{sr} &= \dfrac{\rhod_{t} \rhod_{vr}\sigma_s \sigma_r + b_s b_r \Var(Z_1)}{\sqrt{\sigma_s^2 + S^2 \Var(L_{s,1})} \sqrt{\sigma_r^2 + \Var(L_{r,1})}}, \nonumber \\
\Var(L_{i,1}) &= \Var(Y_{i,1}) + b_i \Var(Z_1), \qquad i \in [s,v,r], \nonumber
\end{align}
\noindent where $\sigma_s, \sigma_v, \sigma_r$ are volatilities of the diffusion processes of $S_t, v_t, R_t$, and $Y_{i,t}$ are the idiosyncratic jump processes. With such a definition all correlations $\rho_{vr}, \rho_{sr}, \rho_t$ could take any value within the interval $[-1,1]$ with no restrictions. As follows from \eqref{corr} $\rho_{vr}$ is constant while $\rho_t, \rho_{sr}$ are stochastic. If the second term in the numerator of $\rho_{sv}$ in \eqref{corr} is much bigger that the first one, this correlation also is constant and determined by the correlation of jumps. If, on the contrary, the second term is small as compared with the first one, $\rho_{sv}$ is stochastic and defined by the correlation of the Brownian motions. A similar argument is valid for $\rho_{sr}$ as well.

\section{Pricing PIDE} \label{sectPIDE}

To price contingent claims, e.g., vanilla or exotic options written on the underlying spot price, a multidimensional (backward) PIDE could be derived by using a standard technique as in \cite{ContTankov}. It describes the evolution of the option price $V$ under risk-neutral measure and reads
\begin{equation} \label{PIDE}
V_{\tau }=[\mathcal{D}+\mathcal{J} - r]V ,
\end{equation}
\noindent where $\tau = T - t$ is the backward time, $T$ is the time to the contract expiration, $r$ is the risk free interest rate, $\mathcal{D}$ is the three-dimensional linear convection-diffusion operator of the form
\begin{align}  \label{Dif}
\mathcal{D} &= F_0 + F_1 + F_2 + F_3, \\
F_1 &= (r - q) S \fp{}{S} + \dfrac{1}{2} \Sigma_s^2 S^{2c} v \sop{}{S} {}{}, \nonumber \\
F_2 &= \kappa_v(t)[\theta_v(t) - v] \fp{}{v} + \dfrac{1}{2} \xi_v^2 v^{2 a} \sop{}{v} , \nonumber
\end{align}
\begin{align*}
F_3 &= \kappa_r(t)[\theta_r(t) - R] \fp{}{R} + \dfrac{1}{2} \xi_r^2 \sop{}{R} , \nonumber \\
F_0 &= \left[\tanh(R) \sigma(S_t) \sigma(v_t) + b_s b_v \Var(Z_1)\right] \cp{}{S}{v} \nonumber \\
&+ \left[\tanh(R) \tilde{\rho}_{vr} \sigma(S_t) \sigma(R_t) + b_s b_r \Var(Z_1)\right] \cp{}{S}{R}
\nonumber \\
&+  \left[ \tilde{\rho}_{vr} \sigma(v_t) \sigma(R_t) + b_v b_r \Var(Z_1)\right] \cp{}{R}{v} \equiv F_{sv} + F_{sr} + F_{vr}. \nonumber
\end{align*}
According to \eqref{corr}
\[ \sigma(S_t) = \Sigma_s(S,t)S^{c} \sqrt{v }, \quad \sigma(v_t) = \xi_v v^{\alpha}, \quad
\sigma(R_t) = \xi_r.  \]

In \eqref{PIDE} $\mathcal{J}$ is the jump operator
\begin{align}  \label{Jump}
\mathcal{J}V &= \int_{-\infty}^{\infty}\int_{-\infty}^{\infty}\int_{-\infty}^{\infty}\Bigg[V(x_s + y_s, x_v + y_v, x_r + y_r, \tau) - V(x_s,x_v,x_r, \tau) \\
&- \sum_{\chi \in [s,v]} (e^{y_\chi} - 1)\fp{V(x_s,x_v,x_r,\tau)}{x_\chi}
- y_r\fp{V(x_s,x_v,x_r,\tau)}{x_r}\Bigg] \mu(d y_s d y_v dy_r), \nonumber
\end{align}
\noindent where $\mu(d y_s d y_v dy_r)$ is the three-dimensional L{\'e}vy measure, and $x_s = \log S/S_0$,  $x_v = \log v/v_0, \ x_r = R$.

This PIDE has to be solved subject to the boundary and terminal conditions. Consider first a pure diffusion case (no jumps). For the FX European vanilla options the terminal condition reads
\begin{equation}
V(S,v,R,T) = P(S),
\end{equation}
\noindent where $P(S)$ is the vanilla option payoff as defined by the corresponding contract (Call or Put). For the double barrier options considered in this paper, the payoff reads
\begin{equation} \label{barPayoff}
V(S,v,R,T) = P(S) \mathbf{1}_{\tau_H > T}\mathbf{1}_{\tau_L > T},
\end{equation}
\noindent where $\tau_H$ is the first time the spot price hits the upper barrier $H$, and $\tau_L$ is the first time the spot price hits the lower barrier $L$. There also exist contracts where some rebate is paid when the spot price hits either the upper or lower barriers. So \eqref{barPayoff} can be easily extended to this case.

The boundary conditions also depend on the type of the options under consideration. In this paper we consider European vanilla and double barrier options. The boundary conditions for those differ only by the conditions at the boundaries of the spot price domain. Namely, we set
\begin{equation}
V(L,v,R,t) = V(H,v,R,t) = 0,
\end{equation}
\noindent for the double barrier options, and
\begin{equation}
V(0,v,R,t) = 0, \quad \fp{V(S,v,R,t)}{S}\Big|_{S \to \infty} = 1
\end{equation}
\noindent for the vanilla Call options.

For the instantaneous variance $v_t$ the form of the boundary conditions depends on the exponent $a$ in \eqref{sde}. For instance, when $a=1/2$, i.e. the diffusion part of $v$ is the CIR process, it is known that no boundary condition is required at $v=0$ if the Feller condition $2 \kappa_v \theta_v > \xi_v^2$ holds. So in this case the equation \eqref{PIDE} with $v=0$ is used as the boundary condition. More general, if $a, \kappa_v, \xi_v, \theta_v$ are chosen in such a way that the process $v$ never reaches the origin, no boundary condition is required. Otherwise, either the reflection or absorbing boundary conditions can be set at $v=0$. At $v \to \infty$ the boundary condition is $\partial V/\partial v \to 0$. Alternatively, it can be written $V(S,v, R, t)\|_{v\to \infty} = S$. In more detail, see, e.g., discussions in \cite{ItkinCarrBarrierR3,Hout3D,Graaf2012}.

In the presence of jumps these conditions should be extended as follows. Suppose we want to use finite-difference method to solve the above PIDE and construct a jump grid, which is a superset of the finite-difference grid used to solve the diffusion equation (i.e. when $\mathcal{J} = 0$, see \cite{Itkin2014}). Then these boundary conditions should be set on this jump grid as well as at the boundaries of the diffusion domain.

Boundary conditions for $R_t$ should be set at $R \to \pm \infty$ which corresponds to $\rho_t \to 1$ and $\rho_t \to -1$. In the pure diffusion case (no jumps) at $\rhod_t = -1$ from \eqref{posSemiDef} we obtain $\rhod_{sr} = \rhod_{vr} = 0$. Thus, in words, this is the case when the spot price is perfectly anti-correlated with the instantaneous variance, while both are independent of their stochastic correlation, which is 1. As follows from \eqref{sde}, the spot price dynamics then is defined by a local volatility model, \cite{Dupire:94,derman/kani:94}, with the local volatility function $\Sigma_{s}(S_t,t) S_t^c\sqrt{v_t}$. Therefore, no boundary condition is necessary at $R \to -\infty$, since the PDE in \eqref{PIDE} will itself serves as the boundary condition at this end. At $\rhod_t = 1$, we have $\rhod_{sr} = \rho_{vr}$. Again, no boundary condition is necessary at $R \to \infty$, and the PDE in \eqref{PIDE} itself serve as the boundary condition with those values of correlations substituted into it.

When the jumps are taken into account, perfect correlation $\rho_l = 1$ is achieved when the loading factors $b_s, b_v$ both tend to infinity. Hence, contributions to \eqref{corr} from both the diffusion correlation and from the idiosyncratic parts are very small. It then follows from \eqref{corr} that
\[ \rho_{sr} = \rho_{vr} = b_r \sqrt{\dfrac{\Var(Z(1))}{\sigma_r^2 + \Var(L_{r,1}) }}. \]
In addition, all terms in the definition of $F_0$ in \eqref{Dif}  tend to infinity when $b_s \to \infty, b_v \to \infty$. Therefore, we must set
\[ \cp{V}{S}{R}  = \cp{V}{R}{v} = \cp{V}{S}{v} = 0. \]
\noindent Hence, similar to the pure diffusion case, we can assume that $V = V(S,\tau)$ at $R \to \infty$. Therefore, \eqref{PIDE} in this limit again transforms to that for the local volatility model, and this equation is just the boundary condition at $R \to \infty$.

By a similar argument, $V = V(S,\tau)$ at $R \to -\infty$. Therefore, substituting this into the PIDE, we obtain the equation which is the boundary condition at this end.

\section{Forward PIDE \label{sect:Forward}}

It is well-known that the backward PDE/PIDE is helpful for pricing options, e.g., with the same strike and maturity but various initial spot prices. However, when calibrating some model to the options market data, we need to simultaneously price multiple options with different strikes and maturities but the same initial spot price. For doing so, a forward Kolmogorov equation can be used. Forward PIDEs can be derived using techniques proposed in \cite{ContBentata2010,AA2000jd,CarrPIDE2006}. Alternatively, we can exploit the approach which is discussed in \cite{Itkin2014b} and for the financial literature is originated by \cite{Lipton2001,Lipton2002}, see a literature survey in \cite{Itkin2014b}. Briefly, the idea is as follow.

Instead of a continuous forward equation we consider a discretized Kolmogorov equation from the very beginning. In other words, this is equivalent to the discrete Markov chain defined at space states corresponding to some finite-difference grid $\mathbf{G}$. Consider first only the stochastic processes with no jumps. It is known (see, e.g., \cite{Goodman2004}) that the forward Kolmogorov equation for the density of the underlying process $X(t)$ can be written in the form
\begin{equation} \label{FKE}
\frac{ \partial p}{ \partial T}  =   p \mathcal{A},
\end{equation}
\noindent where $ p({\bf s},t)$ is a discrete density function (i.e., a vector of probabilities that $X(t)$ is in the corresponding state ${\bf s}$ at time $t$), $T$ is the expiration time, and the generator $\mathcal{A}$ has a discrete representation on $\mathbf{G}$ as a matrix of transition probabilities between the states. For a given interval $ [t, t + \Delta t] $ where the generator $ \mathcal{A} $ does not depend on time,
the solution of the forward equation is
\begin{equation} \label{FKEsol}
p({\bf s},t+\Delta t) = p({\bf s},t)e^{\Delta t \mathcal{A}}.
\end{equation}
Therefore, to compute the option price $V = e^{- r T} E_Q[P(X(T))]$, where $P(X(t))$ is the payoff at time $t$ and $E_Q$ is an expectation under the risk-neutral measure $Q$, we start with a given $p(0)$, evaluate $p(T)$ by solving \eqref{FKE},\footnote{If $\mathcal{A} = \mathcal{A}(t)$, we solve in multiple steps in time $\Delta t$ by using, for instance, a piecewise linear approximation of $\mathcal{A}$ and the solution in \eqref{FKEsol}.} and finally evaluate $e^{-r T}\int_0^\infty p(s,T) P(s,T) ds$.

On the other hand, one can start with a backward Kolmogorov equation for the option price
\begin{equation} \label{BKE}
\fp{V}{t} + \mathcal{A} V  = rV,
\end{equation}
\noindent which by change of variables $\tau = T-t$ could be transformed to
\begin{equation} \label{BKE1}
\frac{ \partial V}{ \partial \tau}  =   \mathcal{B} V - rV.
\end{equation}
It is also well-known that $\mathcal{B} = \mathcal{A}^\top$; see, e.g., \cite{Goodman2004}.

Based on that, constructing the transposed discrete operator $\mathcal{A}^\top$ on a grid $\mathbf{G}$ is fairly straightforward. However, extending this idea  to modern finite-difference schemes of higher-order approximation (see, e.g., \cite{HoutWelfert2007} and references therein) is a greater obstacle.  This obstacle is especially challenging because these schemes by nature do multiple fractional steps to accomplish the final solution, and because an explicit form of the generator $\mathcal{A}$, that is obtained by applying all the steps, is not obvious. Finally, including jumps in consideration makes this problem even harder.

This problem is solved in \cite{Itkin2014b} where a splitting finite-difference scheme for the 2D forward equation is constructed which is fully consistent in the above-mentioned sense with the backward counterpart. For the latter, two popular scheme are chosen, namely: Hundsdorfer and Verwer (HV) and a modified Craig-Sneyd scheme, \cite{HoutWelfert2007}. Also a consistent forward scheme is proposed in \cite{Itkin2014b} for solving the corresponding PIDEs.

\section{Finite-difference scheme for a forward PIDE \label{sect:FD}}

To solve the forward PIDE, we again use the approach of \cite{Itkin2014b}. It consists of few steps

\subsection{Splitting of diffusion and jumps}

We use the Strang splitting scheme, \cite{Strang1968} which provides  the second-order of approximation in a temporal step. For the backward equation it reads
\footnote{In practical implementations, when solving the backward equation, the term $- r V$ is usually included into the operator $\mathcal{D}$ by splitting it into three equal parts $-rV/3,$ and adding each part to $F_1, F_2$ and $F_3$, correspondingly.}
\begin{equation} \label{bwSplit}
V(x,v,R,\tau+\Delta \tau) = e^{\frac{\Delta \tau}{2} \mathcal{D} } e^{\Delta \tau \mathcal{J}}
e^{\frac{\Delta \tau}{2} \mathcal{D} } V(x,v,R,\tau) + O(\dtau^2),
\end{equation}
\noindent which could be re-written as a set of fractional steps:
\begin{align} \label{splitFin}
V^{(1)}(x,v,R,\tau) &= e^{\frac{\Delta \tau}{2} \mathcal{D} } V(x,v,R,\tau), \\
V^{(2)}(x,v,R,\tau) &= e^{\Delta \tau \mathcal{J}} V^{(1)}(x,v,R,\tau) \nonumber, \\
V(x,v,R,\tau+\Delta \tau) &= e^{\frac{\Delta \tau}{2} \mathcal{D} } V^{(2)}(x,v,R,\tau).  \nonumber
\end{align}
Thus, instead of an unsteady PIDE, we obtain one PIDE with no drift and diffusion (the second equation in \eqref{splitFin}) and two unsteady PDEs (the first and third ones in \eqref{splitFin}).

To produce a consistent discrete forward equation, we use the transposed evolutionary operator, which results in the scheme
\begin{equation} \label{fwSplit}
p(x,v,R,t+\Delta t) = e^{\frac{\Delta t}{2} \mathcal{D}^\top } e^{\Delta t \mathcal{J}^\top}
e^{\frac{\Delta t}{2} \mathcal{D}^\top} p(x,v,R,t),
\end{equation}
\noindent assuming that $\dtau = \Delta t$. The fractional steps representation of this scheme is
\begin{align} \label{fwSplitFin}
p^{(1)}(x,v,R,t) &= e^{\frac{\Delta t}{2} \mathcal{D}^\top } p(x,v,R,t), \\
p^{(2)}(x,v,R,t) &= e^{\Delta t \mathcal{J}^\top} p^{(1)}(x,v,R,t) \nonumber, \\
p(x,v,R,t+\Delta t) &= e^{\frac{\Delta t}{2} \mathcal{D}^\top } p^{(2)}(x,v,R,t).  \nonumber
\end{align}

\subsection{Splitting scheme for the pure diffusion equation}

Here we consider the solution of the equation
\begin{equation}
p(x,v,R,t+\Delta t) = e^{\Delta t \mathcal{D}^\top } p(x,v,R,t),
\end{equation}
\noindent which is required at the first and third steps of \eqref{fwSplitFin}. For doing that in \cite{Itkin2014b} two forward finite-difference schemes are proposed for the 2D case. The first scheme, which is consistent with the HV scheme is
\begin{align} \label{fwAlgo}
M^\top_2 Y_0 &= p_{n-1}, \qquad M^\top_1 Y_1 = Y_0, \\
\widetilde{Y_0} &= V_{n-1} + \Delta t \left(c_1^n Y_1 - c_2^n Y_0\right), \nn \\
M^\top_2 \widetilde{Y}_1 &= \widetilde{Y}_0, \qquad M^\top_1 \widetilde{Y}_2 = \widetilde{Y}_1, \nn \\
p_n &= \widetilde{Y}_2 + \Delta t \left[c_3^{n-1} \Delta \widetilde{Y}_2 - c_2^{n-1} \Delta \widetilde{Y}_1 + c_0^{n-1} Y_1\right], \nn
\end{align}
\noindent where the subindex $n$ marks the $n$-th step in time, $M_i \equiv I - \theta\Delta t F^n_i$, $F_i^n = F_i(t_n)$, and $\Delta \widetilde{Y}_i = \widetilde{Y}_i - Y_{i-1}, \ i=1,2$, and $\theta$ is the parameter of the HV scheme.

The forward analog for another popular backward finite-difference scheme---a modified Craig-Sneyd (MCS) scheme, see \cite{HoutFoulon2010}, can also be found in \cite{Itkin2014b}.

Despite both these schemes are unconditionally stable, there exists a delicate issue regarding the positivity of the solution.  As mentioned in \cite{Itkin2014b}, for steps 1, 2, and 4 in \eqref{fwAlgo} this is guaranteed if both $M_1^\top$ and  $M_2^\top$ are M-matrices; see \cite{BermanPlemmons1994}. To achieve this, an appropriate (upward) approximation of the drift (convection) term has to be chosen, which is often discussed in the literature; see \cite{HoutFoulon2010} and references therein.

For steps 3 and 5, a possible issue is that the central difference approximation of the second order for the mixed derivatives doesn't preserves positivity of the solution. Various attempts are made in the literature to replace this approximation with a reduced scheme which uses a 7 point stencil, again see the discussion in \cite{Itkin2014b}. Lack of positivity usually gives rise to instability unless a very small spatial step is used which is impractical. This is less important for the backward scheme because the explicit step is followed by some number of implicit steps which significantly damp possible instabilities (so the resulting solution is non-negative). But for the forward  scheme the explicit step is the last one, and so this could be critical.

To resolve this problem, in \cite{Itkin3D} it is proposed to replace the explicit step for the mixed derivatives with the implicit one. It was proved there that the scheme proposed for calculating of the mixed derivative terms guarantees the positivity of the solution, while the whole scheme becomes fully implicit. Below we describe this scheme for the backward equation while for the forward one it is exactly same if one changes the notation from $V$ to $p$, and from $\tau$ to $t$, and also replaces all matrices with their transposes.

\subsubsection{Fully implicit scheme for the mixed derivatives}

As applied to a 3D case, the idea is to represent the explicit step of the HV scheme \eqref{HV3D} as
\begin{equation} \label{expEq}
Y_0 = V_{n-1} + \dtau F^{n-1} V_{n-1} = (1 + \dtau F^{n-1}) V_{n-1}.
\end{equation}
\noindent where $F = \sum_j F_j, \ j \in [0,3]\cap \mathbb{Z}$. Now observe, that the rhs of \eqref{expEq} is a \pade approximation (0,1) of the equation
\begin{equation} \label{pdePade}
\fp{V(\tau)}{\tau} = [F_0(\tau) + F_1(\tau) + F_2(\tau) + F_3(\tau)]V(\tau).
\end{equation}
The solution of this equation can be formally written as
\begin{align} \label{twoStepI}
V(\tau) &= \exp\left\{\Delta \tau \left[F_0(\tau_{n-1}) + F_1(\tau_{n-1}) +
F_2(\tau_{n-1}) + F_3(\tau_{n-1})\right]\right\}V(\tau_{n-1}) \\
&= e^{\Delta \tau F_0(\tau_{n-1})}e^{\Delta \tau F_1(\tau_{n-1})}
e^{\Delta \tau F_2(\tau_{n-1})}e^{\Delta \tau F_3(\tau_{n-1})}V(\tau_{n-1})
+ O(\Delta \tau). \nonumber
\end{align}

Alternatively, a \pade approximation (1,0) can also be applied to all exponentials in \eqref{twoStepI} providing same order of approximation in $\Delta \tau$ but making all steps implicit. Namely, this  results into the following splitting scheme of the solution of \eqref{pdePade}:
\begin{align} \label{impl1}
V^0(\tau_{n-1}) &= e^{\Delta \tau (F_{Sv}^{n-1} + F_{Sr}^{n-1} + F_{vr}^{n-1})}V(\tau_{n-1}), \\
[1 - \Delta \tau F_1(\tau)]V^1(\tau_{n-1}) &= V^0(\tau_{n-1}), \nonumber \\
[1 - \Delta \tau F_2(\tau)]V^2(\tau_{n-1}) &= V^1(\tau_{n-1}), \nonumber \\
[1 - \Delta \tau F_3(\tau)]V(\tau) &= V^2(\tau_{n-1}). \nonumber
\end{align}
Steps 2-4 in \eqref{impl1}) are the one-dimensional implicit steps, so we already know how to solve these equations.

By construction of the HV scheme, the first step in \eqref{impl1} has to be solved with the accuracy  $O(\dtau)$, again see discussion in \cite{Itkin3D}.  With this accuracy, the previous equation can be factorized as
\begin{equation*} \label{matrixSol}
V^0(\tau_{n-1}) = e^{\Delta \tau F_{Sv}^{n-1}} e^{\Delta \tau F_{Sr}^{n-1}} e^{\Delta \tau F_{vr}^{n-1}}V(\tau_{n-1}),
\end{equation*}
\noindent or using splitting
\begin{align} \label{expSp2}
V^{(1)} &= e^{\Delta \tau F_{Sv}^{n-1}}V(\tau_{n-1}), \\
V^{(2)} &= e^{\Delta \tau F_{Sr}^{n-1}}V^{(1)}, \nonumber \\
V^0(\tau_{n-1}) &= e^{\Delta \tau F_{vr}^{n-1}}V^{(2)}. \nonumber
\end{align}
The order of the splitting steps usually doesn't matter.

For illustration, let us consider a mixed derivative term $F_{Sv}$ in \eqref{expSp2}. We start with replacing the explicit step used in the HV and MCS schemes for the mixed derivative term
\begin{align} \label{expSpl}
V(\tau + \Delta \tau) &= e^{\Delta \tau F_{Sv}(\tau)}V(\tau)
\end{align}
by the implicit representation, which has the same order of approximation (first order in time)
\begin{equation} \label{trick1}
\left[1 - \Delta \tau \rho_{s,v} W(S) W(v) \triangledown_S \triangledown_v \right] V(\tau+\dtau) = V(\tau),
\end{equation}
\noindent where $W(S) = \sqrt{\Var(S_t)}, W(v) = \sqrt{\Var(v_t)}$.

We re-write this equation in the form
\begin{align} \label{trickNeg}
\Big(P &- \sdt \rho_{s,v} W(S)\triangledown_S \Big)\Big(Q + \sdt W(v) \triangledown_v\Big) V(\tau+\dtau) \\
        &= V(\tau) + \left[(P Q-1) - Q\sdt \rho_{s,v} W(S)\triangledown_S + P \sdt W(v) \triangledown_v  \right]V(\tau+\dtau), \nonumber
\end{align}
\noindent where $P,Q,$ are some positive numbers which have to be chosen based on some conditions, e.g., to provide diagonal dominance of the matrices in the parentheses in the lhs of \eqref{trickNeg}.

As shown in \cite{Itkin3D}, \eqref{trickNeg} can be solved using fixed-point Picard iterations. One can start with setting $ V(\tau+\dtau) = V(\tau)$ in the rhs of \eqref{trickNeg}, then solve sequentially two systems of equations
\begin{align} \label{spl2}
\Big(Q &+ \sdt W(v) \triangledown_v\Big) V^* = V(\tau) \nonumber \\
&+ \Big[P Q -1  - Q\sdt \rho_{s,v} W(S)\triangledown_S + P \sdt W(v)\triangledown_v  \Big]V^k, \nonumber \\
\Big(P &- \sdt \rho_{s,v} W(S) \triangledown_S \Big)V^{k+1} = V^*.
\end{align}
Here $V^k$ is the value of $V(\tau+\dtau)$ at the $k$-th iteration.

The following main theorem provides a necessary discretization of the above equation.
\begin{proposition}[Proposition 2.3 from \cite{Itkin3D}] \label{propPos2}
Let us assume $\rho_{s,v} \ge 0$, and approximate the lhs of \eqref{spl2} using the following finite-difference scheme:
\begin{align} \label{fd13}
\Big(Q I_v + \sdt W(v) A^B_{2,v}\Big) V^*
        &= \alpha^+_2 V(\tau) - V^k, \\
\Big(P I_S - \sdt \rho_{s,v} W(S) A^F_{2,S} \Big)V^{k+1} &= V^*, \nonumber      \\
\alpha^{+}_2 = (PQ + 1)I &- Q\sdt \rho_{s,v} W(S) A_{2,S}^B + P\sdt W(v) A_{2,v}^F. \nonumber
\end{align}
Then this scheme is unconditionally stable in time step $\Delta \tau$, approximates \eqref{spl2} with $O(\max(h^2_S,h^2_v))$ and preserves positivity of the vector $V(x,\tau)$ if $Q = \beta \sdt/h_v, \ P = \beta \sdt/h_S$, where $h_v, h_S$ are the grid space steps correspondingly in $v$ and $S$ directions, and the coefficient $\beta$ must be chosen to obey the condition:
\[
\beta > \dfrac{3}{2}\max_{S,v}[W(v) + \rho_{s,v} W(S)].
\]
The scheme \eqref{fd13} has a linear complexity in each direction.
\end{proposition}
For the proof, see Appendix~B in \cite{Itkin3D}. Here we define a one-sided second order approximations to the first derivatives: {\it backward} approximation  $A^B_{2,x}: \ A^B_{2,x} V(x) = [ 3 V(x) - 4 V(x-h) + V(x-2h)]/(2 h)$, and {\it forward} approximation $A^F_{2,x}: \ A^F_{2,x} V(x) = [ -3 V(x) + 4 V(x+h) - V(x+2h)]/(2 h)$. Also $I_x$ denotes a unit matrix. All these definitions assume that we work on a uniform grid with step $h$, however this could be easily generalized for the non-uniform grid as well, see, e.g., \cite{HoutFoulon2010}.

A practical choice of the coefficient $\beta$ is discussed in \cite{Itkin3D}. It is also shown there that in case $\rho_{s,v} < 0$ the same Theorem holds, if in the terms containing $\rho_{s,v}$ we replace $A^F_{2,S}$ with $A^B_{2,S}$ and vice versa.

\subsubsection{Presence of jumps}

When jumps are taken into account, and hence the loading factors $b_i \ne 0, \ i \in [s,v,r]$, according to \eqref{Dif} all mixed derivatives operators will include an extra term. For instance, \eqref{expSpl}
will now read
\begin{align} \label{expSplJump}
V(\tau + \Delta \tau) &= e^{\Delta \tau \left[F_{Sv}(\tau) + b_s b_v \Var(Z_1)\right] }V(\tau) \\
&=  e^{\Delta \tau F_{Sv}(\tau)} e^{\dtau b_s b_v \Var(Z_1)} V(\tau) =
e^{\Delta \tau F_{Sv}(\tau)} \widetilde{V}(\tau), \nonumber \\
\widetilde{V}(\tau) &= e^{\dtau b_s b_v \Var(Z_1)} V(\tau), \nonumber
\end{align}
\noindent where $F_{Sv}(\tau)$ represents the diffusion part of the covariance. Thus, this case is reduced to the previous one with no jumps, since by assumption, $b_i$ and $\Var(Z_1)$ are deterministic.

\subsubsection{Fully implicit scheme for the forward equation}

Using this idea, we derive a forward scheme for the 3D diffusion equation which is an exact adjoint of the backward HV scheme. This is done similar to the derivation in \cite{Itkin2014} for the 2D case. However,  in the present approach we combine it with the fully implicit scheme for the mixed derivative term. The derivation of the implicit-explicit scheme is given in Appendix. The final result reads
\begin{align} \label{HV3DforwTr1}
M_3^{n,\top} Y_3 &= p_{n-1}, \qquad M_2^{n,\top} Y_2 = Y_3, \qquad M_1^{n,\top} Y_1 = Y_2, \\
\widetilde{Y}_0 &= p_{n-1} - \Delta t \theta \left[ \sum_{i=1}^3 (F_i^{n})^\top Y_i  - \dfrac{1}{2\theta} (F^n)^\top Y_1\right],\nonumber \\
M_3^{n-1,\top} \widetilde{Y}_3 &= \widetilde{Y}_0, \qquad M_2^{n-1,\top} \widetilde{Y}_2 = \widetilde{Y}_3, \qquad
M_1^{n-1,\top} \widetilde{Y}_1 = \widetilde{Y}_2, \nonumber \\
Z_1 &= \left(1 + \Delta t F^{n-1,\top}\right) Y_1, \qquad
Z_2 = \left(1 + \Delta t F^{n-1,\top}\right) \Delta Y_1, \nonumber \\
p_n &= \frac{1}{2}(Z_1 + Y_1) - \Delta t \theta \sum_{i=1}^3 (F_i^{n-1})^\top \Delta Y_i + Z_2. \nonumber
\end{align}
\noindent where $\Delta Y = \widetilde{Y} - Y$

To apply Proposition~\ref{propPos2}, we rewrite the fourth line of \eqref{HV3DforwTr1} in the form
\begin{align} \label{F0tilde}
\widetilde{Y}_0 &= p_{n-1} + Y_2 + Y_3 + \frac{1}{2} Y_1
- \sum_{i=1}^3 L^{n}_i Y_i + \frac{1}{2} L^{n}_0 Y_1,\\
L^n_i &= I + \dtau \theta (F_i^{n})^\top, \quad i=1,2,3, \qquad
L^n_0 = I + \dtau (F^{n})^\top. \nonumber
\end{align}
Since we are constructing a scheme which provides the second order of approximation in $\dtau$, the result of the multiplication $\overline{Y}_i = L^n_i Y_i$ is equivalent up to $O((\Delta t)^2)$ to the solution of the equation
\begin{equation} \label{impl}
M_i^{n,\top} \overline{Y}_i = Y_i,
\end{equation}
\noindent see discussions in \cite{Itkin3D} about \pade approximations and properties of M-matrices. The \eqref{impl} is a pure implicit equation, and its solution is a non-negative vector provided that $M_i^\top$ is an M-matrix, and vector $Y_i$ is also non-negative.

Thus, \eqref{F0tilde} can be re-written as
\begin{align} \label{F0tilde2}
\widetilde{Y}_0 &= p_{n-1} + (Y_2 - \overline{Y}_2) + (Y_3 - \overline{Y}_3) + \left[\frac{1}{2} (Y_1 + Z_0) - \overline{Y}_1\right], \\
Z_0 &= (I + \dtau (F^{n})^\top)Y_1. \nonumber
\end{align}
The equation for $Z_0$  takes a form of \eqref{expEq}. Therefore, Proposition~\ref{propPos2} can be applied to compute the vector $\widetilde{Y}_0$ in a fully implicit manner which preserves the positivity of the solution, and is of second order of approximation in the spatial steps.

In a similar way we can proceed with the last line of \eqref{HV3DforwTr1} which can be represented as
\begin{align} \label{VnT}
p_n &=\frac{1}{2}(Z_1 + Y_1) +Z_2 - \Delta t \theta \sum_{i=1}^3 (F_i^{n-1})^\top \Delta Y_i \\
&= \frac{1}{2}(Z_1 + Y_1) + Z_2 + \sum_{i=1}^3 \Delta Y_i - \sum_{i=1}^3 L_i^{n-1} \Delta Y_i \nonumber
\end{align}
Introduce $\widehat{Y}_i$ as the solution of the equation
\begin{equation} \label{impl11}
M_i^{n,\top} \widehat{Y}_i = \widetilde{Y}_i.
\end{equation}
Observe, that since the solution of \eqref{VnT} has to be obtained with the accuracy up to $O((\Delta t)^2)$, we can replace operators $L^{n-1}_i$ in the second sum of the last line of \eqref{VnT} by $L^{n}_i$, because
$\Delta t F^{n-1}_i = \Delta t F^{n}_i + O((\Delta t)^2)$. For the same reason $Z_1 = Z_0 + O((\Delta t)^2)$. Then
\begin{align} \label{VnT1}
p_n &= \frac{1}{2}(Z_0 + Y_1) + Z_2 + \sum_{i=1}^3 \Delta Y_i + \sum_{i=1}^3 (\overline{Y}_i - \widehat{Y}_i), \\
Z_2 &= [I + \dtau (F^{n-1})^\top)]\widetilde{Y}_1. \nonumber
\end{align}
The last term in this equation, again has the form of \eqref{expEq} and can also be computed using the fully implicit scheme described in Proposition~\ref{propPos2}, if one replaces $A_2^B, A_2^T$ by their transposes. This finalizes the construction of the fully implicit splitting scheme for the 3D advection-diffusion equation which provides second order of approximation in temporal and spatial steps, is unconditionally stable and preserves positivity of the solution.

As shown in \cite{Itkin3D}, this fully implicit scheme allows elimination of first few Rannacher steps as this is usually done in the literature to provide a better stability (see survey, e.g., in \cite{Hout3D}), and provides much better stability of the whole scheme which is important when solving multidimensional problems.

For the sake of convenience, below we collect all the above splitting steps just in one block
\begin{align} \label{HV3DforwFinal}
M_3^{n,\top} Y_3 &= p_{n-1}, \qquad M_2^{n,\top} Y_2 = Y_3, \qquad M_1^{n,\top} Y_1 = Y_2, \\
Z_0 &= (I + \dtau (F^{n})^\top)Y_1. \nonumber \\
M_i^{n,\top} \overline{Y}_i &= Y_i, \ i=1,2,3, \nonumber \\
\widetilde{Y}_0 &= p_{n-1} + (Y_2 - \overline{Y}_2) + (Y_3 - \overline{Y}_3) + \left[\frac{1}{2} (Y_1 + Z_0) - \overline{Y}_1\right], \nonumber \\
M_3^{n-1,\top} \widetilde{Y}_3 &= \widetilde{Y}_0, \qquad M_2^{n-1,\top} \widetilde{Y}_2 = \widetilde{Y}_3, \qquad
M_1^{n-1,\top} \widetilde{Y}_1 = \widetilde{Y}_2, \nonumber \\
Z_2 &= \left(1 + \Delta t F^{n-1,\top}\right) \widetilde{Y}_1, \nonumber \\
M_i^{n,\top} \widehat{Y}_i &= \widetilde{Y}_i, \ i=1,2,3, \nonumber \\
p_n &= Z_2 + \frac{1}{2}(Y_1-Z_0) + \sum_{i=1}^3 (\widetilde{Y}_i - Y_i + \overline{Y}_i - \widehat{Y}_i). \nonumber
\end{align}

Complexity-wise, this scheme requires the solution of 12 implicit systems of linear equations, and 2 explicit equations. The latter steps being treated using the fully-implicit scheme require per one explicit equation: 3 implicit steps for $i=1,2,3$, approximately 2 iterations for each of 3 mixed terms, and for each term solving 2 implicit systems, so totally 15 implicit systems. Overall, the whole solution could be obtained by solving 42 implicit systems\footnote{By re-grouping terms in \eqref{HV3DforwFinal} it actually can be reduced to 36.}. Therefore, the complexity of this approach is about $42\cdot O(N_1 N_2 N_3)$ where $N_i$ is the number of the grid nodes in the $i$-th direction, $i=1,2,3$. So this is about 2.3 times slower than the original HV scheme for the backward equation.

\subsection{Preserving the norm of the forward solution}

Since the discrete solution of the forward Kolmogorov equation is the discrete probability density, it has to obey two important conditions. The first one is the positivity of the solution and was considered in the previous section. The other condition requires that the three-dimensional integral of the density should be equal to 1. In the discrete case this can be replaced with the condition that
\begin{equation} \label{normCond}
\sum_{i=1}^{N_1} \sum_{j=1}^{N_2} \sum_{k=1}^{N_3} p(S_i,v_j,R_k,t) = 1, \qquad \forall t \in [0,T].
\end{equation}
Then the natural question would be: do there exist any conditions that an appropriate finite-difference scheme should obey in order to automatically preserve \eqref{normCond}? In the one-dimensional case the answer to this question is given by the following Proposition:
\begin{proposition} \label{propNorm}
Consider a forward scheme \eqref{HV3DforwFinal}. A typical splitting step of this scheme consists in solving an equation of the type
\[ M_i^\top Y = Y_0, \]
 \noindent where $Y_0$ is a given vector, $Y$ is the unknown vector to be determined, and $M = I - \theta\Delta t F$ is the lhs matrix. The vector $Y_0$ has all positive elements, and its $L_1$ norm is 1, i.e.
 \[ \sum_{k=1}^N |Y_{0,k}| = \sum_{k=1}^N Y_{0,k} = 1. \]
  In other words, $Y_0$ belongs to the probability vectors $\mathbb{P}$, $Y_0 \in \mathbb{P}$, or to the discrete distributions. Then $Y$ is also a discrete distribution, i.e. $Y \in \mathbb{P}$, if discretization of $F$ is such that the matrix $B = \|F\|$ is a valid generator matrix,
  $B \in \mathbb{G}$ (see \cite{ItkinBook} and references therein), and $M$ is an EM- (or an M-)\footnote {Every M-matrix is also an EM-matrix.} matrix.
\end{proposition}

\begin{proof}
Since $B \in \mathbb{G}$,  it has the following property of the generator matrix, \cite{ItkinBook}
\begin{equation} \label{normB}
\sum_{j=1}^N B_{i,j} = 0, \quad \forall i \in [1,N].
\end{equation}
If $B$ satisfies  \eqref{normB}, then obviously
\[ \sum_{j=1}^N M_{i,j} = 1, \quad \forall i \in [1,N]. \]
On the other hand, one can represent $Y$ as  $Y = (M^\top)^{-1} Y_0$. Therefore, to have $Y$ with the same $L_1$ norm as that of $Y_0$, the matrix $(M^\top)^{-1}$ should have its columns to be in $\mathbb{P}$, $(M^\top)^{-1}_{:,j} \in \mathbb{P} \ \forall j \in [1,N]$.

Observe, that since $M$ is an EM-matrix, $M^\top$ is also an EM-matrix, and $(M^\top)^{-1}$ is a positive matrix, again see \cite{ItkinBook}. Now the result that $(M^\top)^{-1}_{:,j} \in \mathbb{P} \ \forall j \in [1,N]$ follows from the fact that if $Y_0 \in \mathbb{P}$, $(M^\top)^{-1} Y_0$ is a convex combination of the columns of $(M^\top)^{-1}$ with coefficients given by the entries of $Y_0$. Each column of $(M^\top)^{-1}$ must be in $\mathbb{P}$, and $\mathbb{P}$ is a convex set.

Further we need the following Lemma:
\begin{lemma} \label{lemmaNorm}
Suppose we are given an invertible matrix $M^\top$ with real entries with the sum of every column is 1. Then, for the matrix $(M^\top)^{-1}$ the sum of the elements in each column is 1.
\end{lemma}
\begin{proof}
The statement of this Lemma is obviously equivalent to the following statement. Given an invertible matrix $M$ with real entries with the sum of every row is 1, for the matrix $M^{-1}$ the sum of the elements in each row is also 1. To prove this, denote $M^{-1} = D$. Then we have
\begin{align*}
\sum_j M_{i,j} &= 1, \qquad \forall j,i \in [1,N], \\
\sum_j D_{i,j} M_{j,k}  &= \delta_{i,k}, \qquad \forall k,i \in [1,N],
\end{align*}
\noindent and thus
\[ \sum_j D_{i,j} = \sum_j D_{i,j} \left(\sum_k M_{j,k} \right) = \sum_{j,k} D_{i,j} M_{j,k}  = \sum_k \delta_{i,k} = 1. \qquad \blacksquare \]
\end{proof}
Therefore, since by construction $M^\top$ is such a matrix that its every column sums to 1, based on the above Lemma it follows that $(M^\top)^{-1} \in \mathbb{P}$. This proves the statement of the Proposition.
\hfill  $\blacksquare$
\end{proof}

In 3D case the splitting scheme in \eqref{HV3DforwFinal} consists of multiple 1D steps. Therefore, intuitively, it is clear that if all discretizations of matrices $F_i, \ i=0,1,2,3$ are chosen according to the Proposition~\ref{propNorm}, the total scheme preserves the $L_1$ norm of the solution vector. A rigorous prove of this statement will be given elsewhere. This is also confirmed by our numerical experiments.

\subsection{Splitting scheme for the pure jump equation}

For the second (jump) step in \eqref{splitFin} a splitting scheme is proposed in \cite{ItkinLipton2014,Itkin3D}. If, e.g., jumps in $S$ are described by an exponential \LY process, the jump integral admits representation in the form of a pseudo-differential operator
\begin{align}  \label{intGen}
\int_\mathbb{R} \left[ V(x+y,t) \right. & \left. - V(x,t) - (e^y-1)
\fp{V(x,t)}{x} \right] \nu(dy) = \mathcal{J} V(x,t), \\
\mathcal{J} & \equiv \int_\mathbb{R}\left[ \exp \left( y \dfrac{\partial}{
\partial x} \right) - 1 - (e^y-1) \fp{}{x} \right] \nu(dy),  \nonumber
\end{align}
\noindent which is introduced in \cite{ItkinCarr2012Kinky}, and where $x = \log S$. In the definition of the operator $\mathcal{J}$ (which is actually an infinitesimal generator of the jump process), the integral can be formally
computed under some mild assumptions about existence and convergence if one
treats the term $\partial/ \partial x$ as a constant. It is important to emphasize that
\begin{equation}  \label{MGF}
\mathcal{J} = \psi(-i \partial_x) = \mbox{MGF}(\partial_x),
\end{equation}
\noindent where $\psi(u)$ is the characteristic exponent of the jump process, and $\mbox{MGF(u)}$ is the moment generation function corresponding to this characteristic exponent. This directly follows from the L{\'e}vy-Khinchine theorem. This point of view apparently has been pioneered in \cite{Jacob1996}. See also the detailed state-of-the-art surveys in \cite{Jacob2001,BSW2014}. Then, using Proposition~\ref{propBallotta}, the operator $e^{\dtau \mathcal{J}}$ in \eqref{splitFin} can be represented as, \cite{ItkinLipton2014,Itkin3D}
\begin{align} \label{splitting}
e^{\dtau \mathcal{J}} &=
e^{0.5 \Delta \tau \mathcal{J}_s}
e^{0.5 \Delta \tau \mathcal{J}_v}
e^{0.5 \Delta \tau \mathcal{J}_r}
e^{\Delta \tau \mathcal{J}_{123}}
e^{0.5 \Delta \tau \mathcal{J}_r}
e^{0.5 \Delta \tau \mathcal{J}_v}
e^{0.5 \Delta \tau \mathcal{J}_s}, \\
J_\eta &= \phi_\eta(- i \triangledown_\eta), \ J_{123} = \eta_Z(- i \sum_{\eta \in [s,v,r]} b_\eta \triangledown_\eta), \ \triangledown_\eta \equiv \fp{}{\eta}. \nonumber
\end{align}
Thus, this requires a sequential solution of 7 equations at every time step.

As shown in \cite{Itkin2014b}, using this method for the forward equation doesn't bring any problem, because
\begin{equation} \label{splittingT}
\left(e^{\dtau \mathcal{J}}\right)^\top =
e^{0.5 \Delta \tau \mathcal{J}^\top_s}
e^{0.5 \Delta \tau \mathcal{J}^\top_v}
e^{0.5 \Delta \tau \mathcal{J}^\top_r}
e^{\Delta \tau \mathcal{J}^\top_{123}}
e^{0.5 \Delta \tau \mathcal{J}^\top_r}
e^{0.5 \Delta \tau \mathcal{J}^\top_v}
e^{0.5 \Delta \tau \mathcal{J}^\top_s},
\end{equation}
\noindent and there is no issue with computing $\mathcal{J}^\top$. Indeed, if $\mathcal{J}$ is the negative of an M-matrix, the transposed matrix preserves the same property. Also if $\mathcal{J}$ has all negative eigenvalues, the same is true for $\mathcal{J}^\top$. Then the unconditional stability of the scheme and its property of preserving the positivity of the solution follow from Proposition 4.1 in \cite{Itkin2013}.
Also, it is easy to show that the jump equation at the central step of \eqref{splittingT}
\[ e^{\Delta \tau \mathcal{J}^\top_{123}} Z = Y \]
\noindent can be solved with respect to the unknown vector $Z$ by using the same ADI method as in \cite{ItkinLipton2014,Itkin3D}, by simply replacing all matrices $\|\triangledown_\eta\|$ with
$\|\triangledown_\eta\|^\top, \ \eta \in [s,v,r]$.

\subsection{Boundary and initial conditions}

As mentioned in \cite{Itkin2014b}, the boundary conditions for the forward scheme should be consistent with those for the backward scheme. However, these two sets of conditions are not exactly same because for the forward equation the dependent variable is the density, while for the backward equation the dependent variable is the undiscounted option price.

In Section~\ref{sectPIDE} we have already discussed the boundary  and initial (terminal) conditions for the backward equation. For the forward equation the obvious initial condition reads
\[ p(S,v,R)\Big|_{t=0} = \delta(S-S_0)\delta(v-v_0)\delta(R-R_0), \]
\noindent where $\delta(x)$ is the Dirac delta-function.

Setting the boundary conditions is a more delicate issue. In the $S$ domain, we require the density $p(S,v,R,t)$ to vanish at the domain boundaries, i.e. at $S=0$ and $S \to \infty$. The boundary conditions in the $v$-domain are discussed, e.g., in \cite{Lucic2008}. For the Heston model, if the Feller condition $2\kappa_v \theta_v > \xi_v$ holds, the boundary condition at $v=0$ is $\partial p(S,v,R,t)/\partial v = 0$. Otherwise, $p(t,S,0,R) = 0$ should be used as the boundary condition. Therefore, in this paper we use this approach, i.e. if the point $v=0$ is inaccessible based on given values of the parameters $a, \kappa_v, \xi_v, \theta_v$, we use the boundary condition $\partial p(S,v,R,t)/\partial v = 0$ at $v=0$. At $v \to \infty$ we set $p(S,v,R,t) \to 0$.

Also, since by construction the $R$ domain doesn't contain any reflection or absorbing boundaries, the analysis of the option price behavior at $R \to \pm \infty$ given in Section~\ref{sectPIDE} can be applied here as well. Based on that, we conclude that similar to the $S$ domain, we can require the density $p(S,v,R,t)$ to vanish at $R \to \pm \infty$.

\section{Numerical experiments \label{sect:numerics}}

In this section we describe some results of our numerical experiments. Here, for the sake of brevity we provide just a single example where jumps are taken into account, all the other examples deal with a pure diffusion case. We will present more detailed results for the model with jumps elsewhere.

Fo all experiments we solve the forward Kolmogorov equation \eqref{FKE} using the proposed finite-difference fully implicit splitting scheme, and compute the density $p(S,v,R.t)$. Parameters of the model used in these tests are given in Tab.~\ref{Tab1}, where $L,H$ denote the lower and upper barriers, correspondingly. The local volatility function $\sigma(S,t)$ in these tests is always set to 1. Power parameters are taken to be $a=0.5, c=1$. Thus, this setting is analogous to the Heston model, but with the stochastic correlation.
\begin{table}[!htb]
\begin{center}
\scalebox{0.9}{
\begin{tabular}{|c|c|c|c|c|c|c|c|c|c|c|c|c|c|c|c|}
\hline
$T$ & $r_d$ & $r_f$ & $L$ & $H$ & $\kappa_v$ & $\xi_v$ & $\theta_v$ & $\kappa_r$ & $\xi_r$ & $\theta_r$ & $\rho_{vr}$ & $S_0$ & $v_0$ & $\rho_{xy,0}$ \\
\hline
0.5 & 0.02 & 0.01 & 50 & 84.5 & 2 & 0.3 & 0.1 &
 0.3 & 5 & -0.2 & 0.4 & 65 & 0.5 & -0.7  \\
\hline
\end{tabular}
}
\caption{Parameters of the test for pricing an European Call option.}
\label{Tab1}
\end{center}
\end{table}
The finite-difference non-uniform grid is constructed similar to how this was done in \cite{Itkin3D} and includes 101 nodes in $S$, and 81 nodes in $v$ and $R$ directions, so the whole grid is $101 \times 81 \times 81$. The fixed temporal step is 0.01. No Rannacher or any other smoothing scheme is used at the first temporal steps. Parameters $\beta$ for computing each mixed derivatives term by using the fully implicit scheme are taken to be 10. Parameter $\theta$ of the HV scheme according to \cite{Itkin2014b} is 0.3. We use Matlab and run our code at PC with Dual Quad Core Intel(R) Core(TM)i7-4790 CPU 3.60 Ghz. 

Since the proposed scheme is constructed by using the transposed operators (matrices), and a fully implicit scheme for the mixed derivatives, its convergence analysis coincides with that provided in \cite{Itkin3D}. Therefore, in our numerical experiments we are mostly dedicated to the financial interpretation of the results obtained.

\subsection{European options}

In this tests we applied our model to pricing European vanilla options. As a benchmark the Heston model is chosen which is mimicked by setting  $\kappa_r = \xi_r = \rho_{yz} = 0$ in the general setting described in above.

The ATM Call option value in this test at $T=0.5$ computed using the forward equation is 10.7533 while the benchmark value computed using FFT is 10.7564. Thus, for this test the relative accuracy is about 0.02\%.
An elapsed time for computing one step in time is about 4.8 sec.

\afterpage{%
\clearpage\clearpage
\begin{figure}[t]
\begin{center}
\includegraphics[width=\textwidth, height=3in]{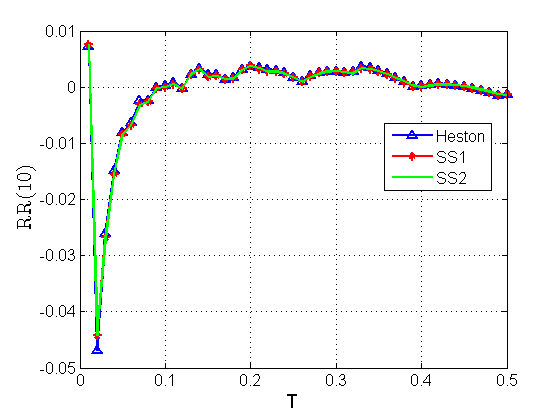}
\caption{RR(10) skew for the Heston and stochastic skew models.}
\label{eurSkew}
\end{center}
\end{figure}

\begin{figure}[!htb]
\begin{center}
\includegraphics[width=\textwidth, height=3in]{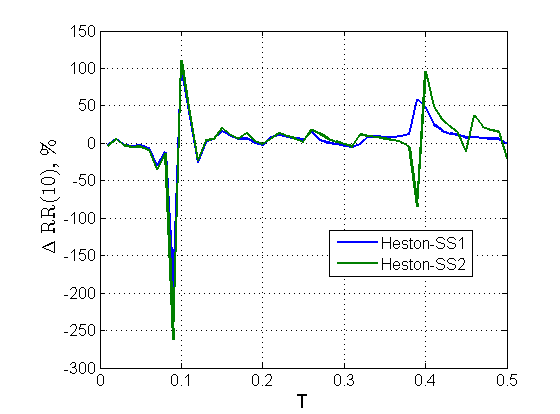}
\caption{Difference in RR(10) skew for the Heston and stochastic skew models.}
\label{eurSkewDif}
\end{center}
\end{figure}
}
Fig.~\ref{eurSkew} presents the RR(10) skew of the European options computed using the proposed model with parameters as in Table~\ref{Tab1} (the SS1 model) and the Heston model. It can be seen that the Heston model also generates stochastic skew. This is the known fact,\footnote{I appreciate our discussion with Peter Carr on the subject.} because the skew is proportional to the spot/InV covariance. Therefore, even if the correlation coefficient is constant, the skew changes with time if the spot variance, or the vol-of-vol, or both change with time (so their product changes with time). 

It is also seen that making the correlation stochastic doesn't contribute much to the magnitude of the skew. In Fig.~\ref{eurSkewDif} the difference in skew between the Heston and SS1 models is presented together with the third case which displays the difference in skew between the Heston model and the stochastic skew model (SS2) where in contrast to Table~\ref{Tab1} we use $\kappa_r = 3, \theta_r = 0., \xi_r = 5, \rho_{yz} = 0.4$.
In other words, in the SS2 case there is no mean-reversion in the dynamics of the correlation coeffcient $\rho_t$. The relative difference in skew as compared with the Heston model could reach +100\% or -250\%, however, this is because the absolute value of the skew is small.

Fig.~\ref{Eur3DXY},\ref{Eur3DXZ},\ref{Eur3DYZ} display the computed density of the European vanilla options at maturity. Since the whole picture is 4D, the results are presented in various planes. In other words, the third dimension is represented as a sequence of graphs corresponding to some discrete values of the third independent variable. It can be seen, that the density in various planes could be bimodal, or even quadro-modal. Such type of the solutions for the density function has been already observed, e.g., in \cite{VerwerDelta2008,Kumar2006}.

\afterpage{%
\clearpage\clearpage

\begin{figure}[t]
\begin{center}
\fbox{\includegraphics[width=\textwidth, height=2.9in]{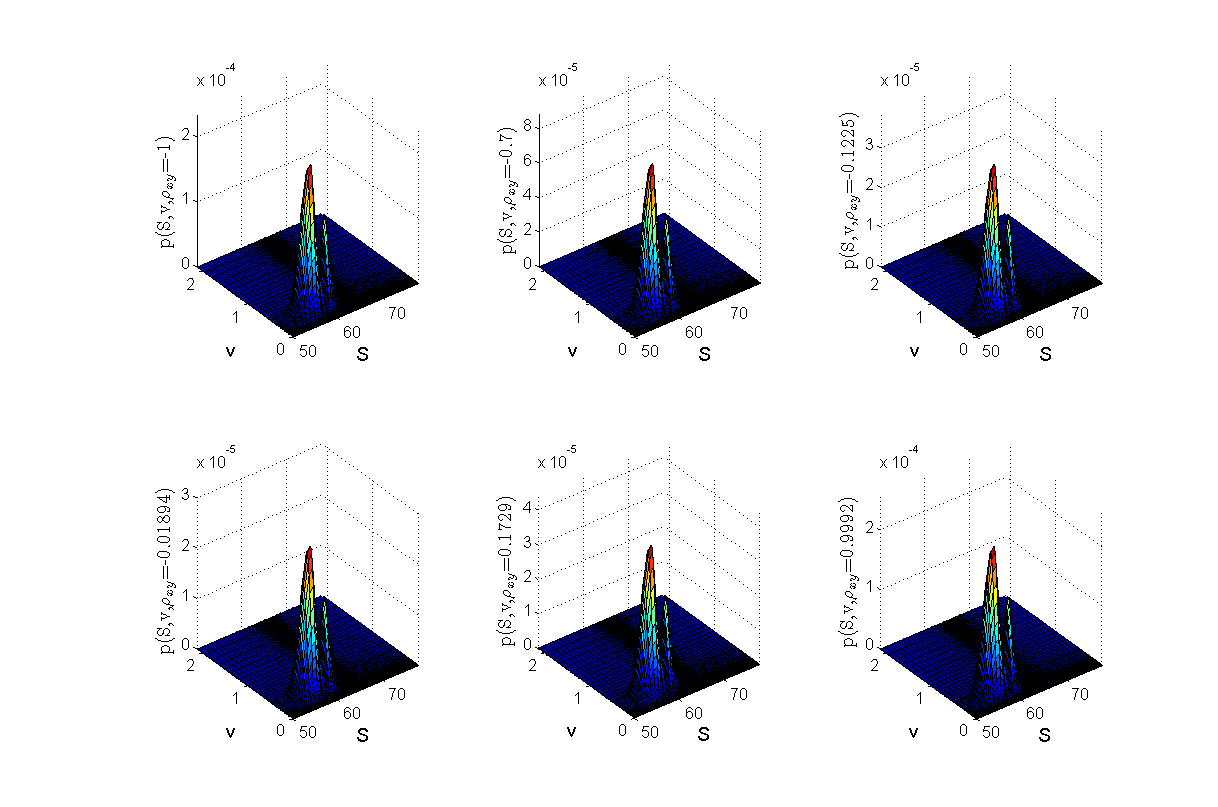}}
\caption{Density of the European vanilla option in $S_0-v_0$ plane at various values of $\rho_{xy,0}$.}
\label{Eur3DXY}
\end{center}
\end{figure}

\begin{figure}[!htb]
\begin{center}
\fbox{\includegraphics[width=\textwidth, height=2.9in]{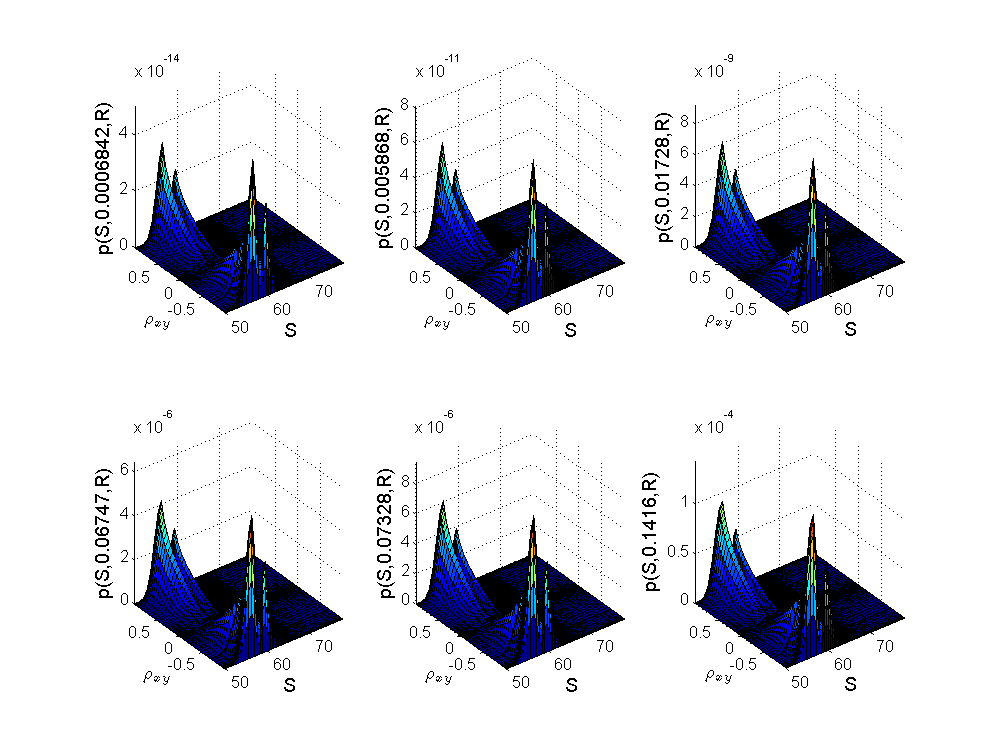}}
\caption{Density of the European vanilla option in $S_0-\rho_{xy,0}$ plane at various values of $v_0$.}
\label{Eur3DXZ}
\end{center}
\end{figure}
}

\afterpage{%
\clearpage\clearpage

\begin{figure}[t]
\begin{center}
\fbox{\includegraphics[width=\textwidth, height=2.9in]{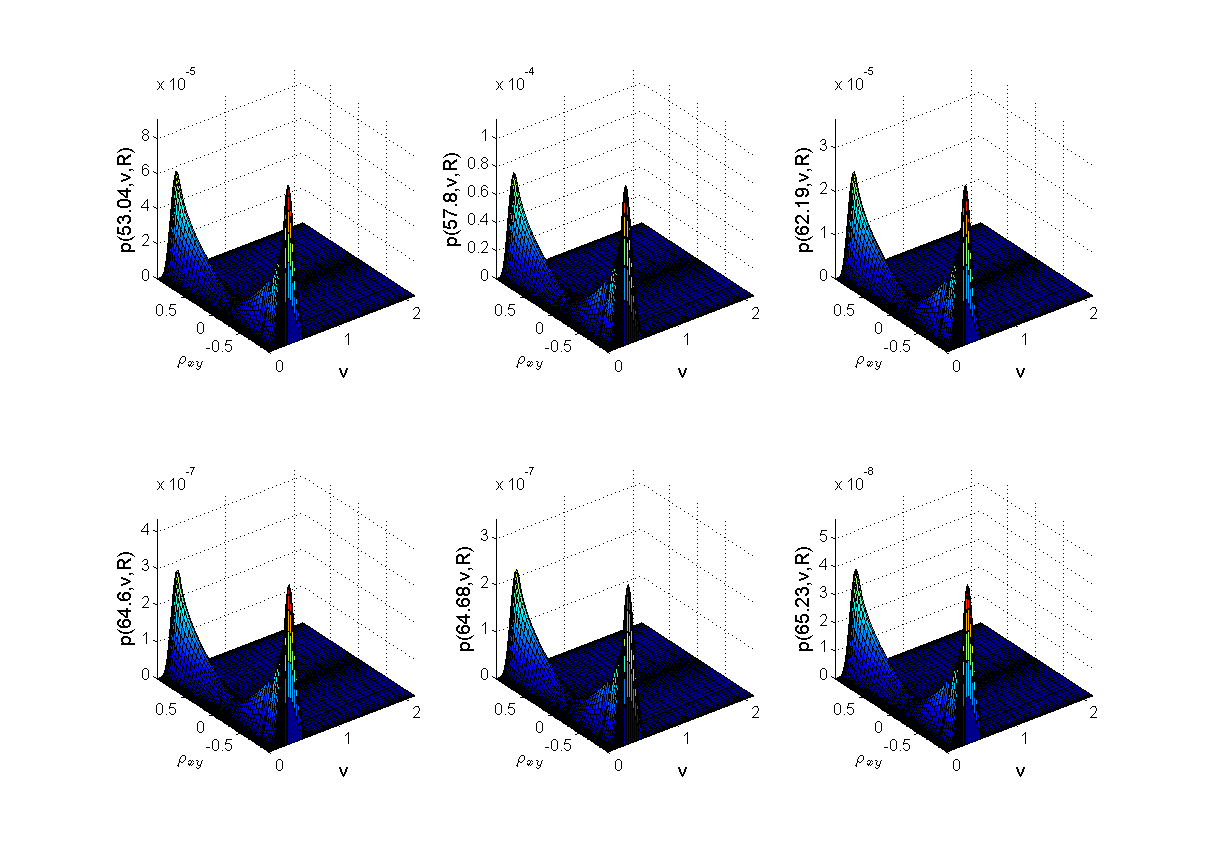}}
\caption{Density of the European vanilla option in $v_0-\rho_{xy,0}$ plane at various values of $S_0$.}
\label{Eur3DYZ}
\end{center}
\end{figure}

\begin{figure}[!htb]
\begin{minipage}{0.45\textwidth}
\begin{center}
\fbox{\includegraphics[width=\textwidth, height=2.2in]{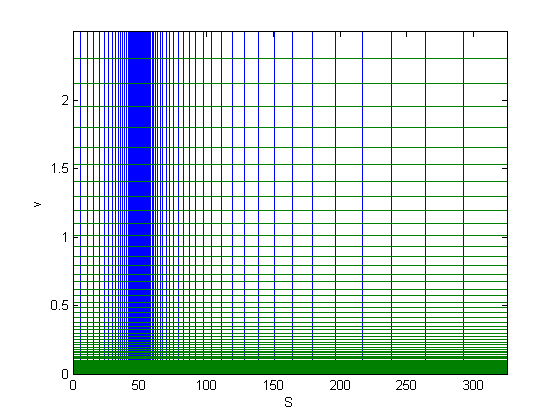}}
\end{center}
\end{minipage}
\hspace{0.1\textwidth}
\begin{minipage}{0.45\textwidth}
\begin{center}
\fbox{\includegraphics[width=\textwidth, height=2.2in]{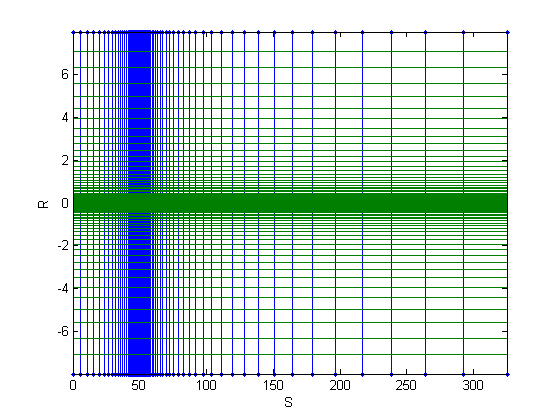}}
\end{center}
\end{minipage}
\caption{Grid for pricing DO options in planes $S_V$ and $S_R$.}
\label{gridDO}
\end{figure}
}

\subsection{Down-and-out barrier options} \label{doVanilla}

It is known, that when the lower barrier $L$ is fixed, for the DO barrier Call option there always  exists a strike corresponding to the option $\delta$ = 0.1 (10 Delta). However, for the Put this stops to be true when $T$ increases. So starting from some maturity $T$, the equation $\partial V(S,K,v,\rho,T,r_d,r_f)/\partial S = 0.1$ has no solution with respect to $K$. Therefore, in practice to trade the skew the barrier options are often set with the barriers being a function of time, \cite{FXSkew2016}. 

In the next test we set a linear time-dependent lower barrier to be $L(t) = L - 20 t$, where $L$ is given in Table~\ref{Tab1}. The corresponding results for the skew of the down-and-out barrier options are presented in Fig.~\ref{doSkew},\ref{doSkewDif}. Here, we again use a non-uniform grid with 100 x 80 x 80 nodes. Projections of the grid in planes $S-v$ and $S-R$ are presented in Fig.~\ref{gridDO}.

It can be seen that assuming the correlation coefficient $\rho_{sv}$ to be stochastic bring some changes in the skew value as compared with the pure stochastic volatility Heston model. At the short time scale these changes are within $\pm$50 bps. Therefore, if the market demonstrates larger variations of the skew, adding jumps to the correlation could make the model more flexible.

\afterpage{%
\clearpage\clearpage

\begin{figure}[t]
\begin{center}
\includegraphics[width=\textwidth, height=2.9in]{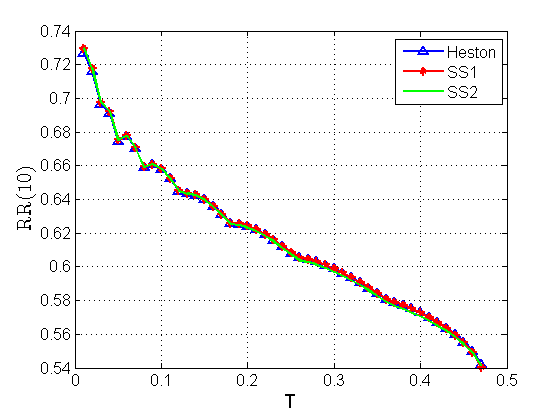}
\caption{RR(10) skew for the Heston and stochastic skew models for Down-and-out barrier options.}
\label{doSkew}
\end{center}
\end{figure}

\begin{figure}[!htb]
\begin{center}
\includegraphics[width=\textwidth, height=2.9in]{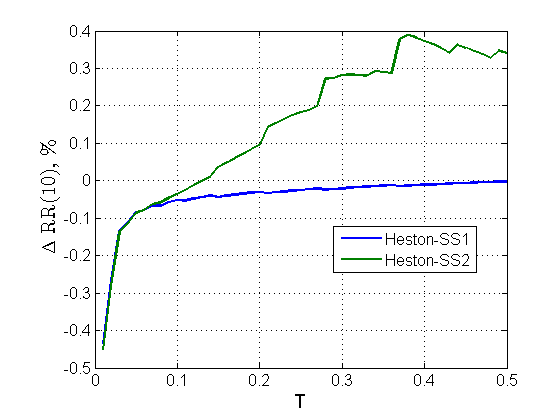}
\caption{Difference in RR(10) skew for the Heston and stochastic skew models for Down-and-out barrier option.}
\label{doSkewDif}
\end{center}
\end{figure}
}

\subsection{Double no-touch options}

In this test we compute prices of double no-touch (DNT) options with fixed the lower and upper barriers which values are given in Tab.~\ref{Tab1}. Again, we use a non-uniform grid 100 x 80 x 80 which projections in planes $S-v$ and $S-R$ are presented in Fig.~\ref{gridDNT}.

\afterpage{%
\clearpage\clearpage

\begin{figure}[t]
\begin{minipage}{0.45\textwidth}
\begin{center}
\fbox{\includegraphics[width=\textwidth, height=2.2in]{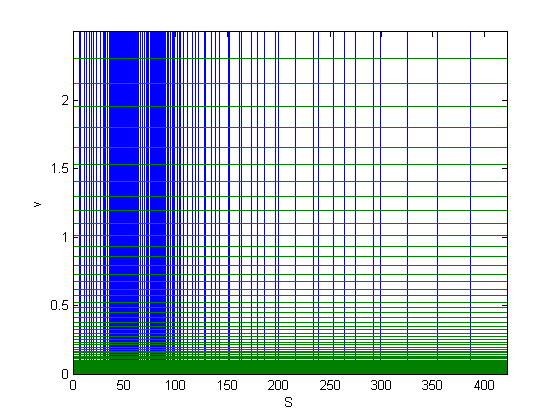}}
\end{center}
\end{minipage}
\hspace{0.1\textwidth}
\begin{minipage}{0.45\textwidth}
\begin{center}
\fbox{\includegraphics[width=\textwidth, height=2.2in]{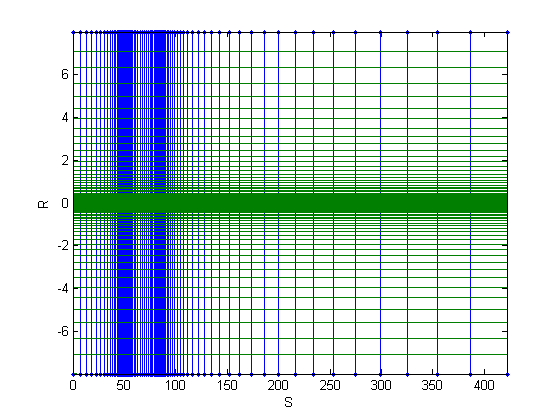}}
\end{center}
\end{minipage}
\caption{Grid for pricing DNT options in the planes $S_V$ and $S_R$.}
\label{gridDNT}
\end{figure}

\begin{figure}[!htb]
\begin{center}
\includegraphics[width=\textwidth, height=2.9in]{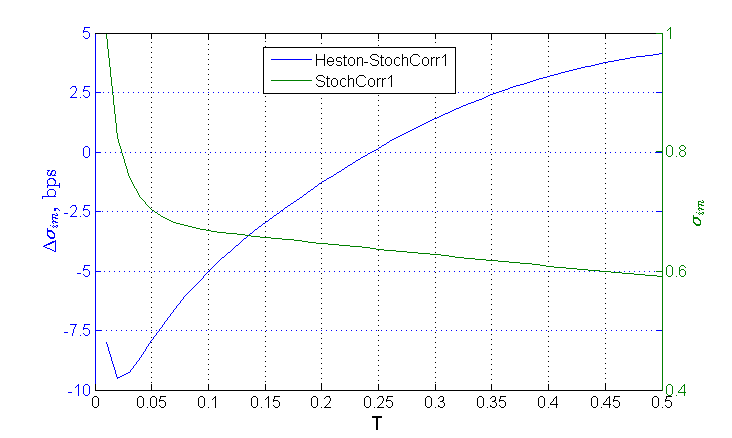}
\caption{Difference in the implied volatility for the Heston and stochastic skew model for double no-touch options with different maturities $T$.}
\label{dntIV}
\end{center}
\end{figure}
}

In Fig.~\ref{dntIV} the implied volatility of the DNT options is presented as a function of time to maturity $T$ computed for the Heston and stochastic correlation models. At the given time scale the difference between two models is about 5-10 bps and also changes the sign with time.

\subsection{Down-and-out barrier options (with jumps)}

This test uses the setting of Section~\ref{doVanilla}, but now includes jumps into consideration. We represent all jumps - common and idiosyncratic - by using the Kou double exponential model, \cite{Kou2004}, similar to how this was done in \cite{ItkinLipton2014}.

The \LY density of the double-exponential jumps is, \cite{ContTankov}:
\begin{equation}  \label{Kou}
\nu (dx) = \varphi\left[ p \theta_1 e^{-\theta_1 x} \mathbf{1}_{x \ge 0} +
(1-p) \theta_2 e^{\theta_2 x} \mathbf{1}_{x < 0} \right] dx,
\end{equation}
where $\varphi$ is the jumps intensity, $\theta_1 > 1$, $\theta_2 > 0$, $1 > p > 0$; the first condition was
imposed to ensure that the underlying asset price has a finite expectation.

Accordingly, for this process $\Var(L_{i,1})$ in \eqref{corr} in the explicit form reads
\begin{equation} \label{Var1}
\Var(L_{i,1}) = \varphi_i \left[ \dfrac{p_i}{\theta_{1,i}^2} +  \dfrac{1-p_i}{\theta_{2,i}^2}\right], \qquad
i \in [S,v,r,Z].
\end{equation}

Parameters of the jumps model for this particular test are presented in Table~\ref{TabJumpParam}

\afterpage{%
\clearpage\clearpage
\begin{table}[t]
\begin{center}
\begin{tabular}{|c|c|c|c|c|c|}
\hline
Process & $\varphi$ & $p$ & $\theta_1$ & $\theta_2$ & $b$ \cr
\hline
$L_{s,t}$ & 0.3 & 0.3 & 3 & 4 & 3 \cr
\hline
$L_{v,t}$ & 0.3 & 0.4 & 2 &  3 & 2 \cr
\hline
$L_{r,t}$ & 0.3 & 0.6 & 1.5 & 2 & 5 \cr
\hline
$Z_t$ 	  & 0.3 & 0.3 & 3 & 3.5 & - \cr
\hline
\end{tabular}
\caption{Parameters of the 3D jump models.}
\label{TabJumpParam}
\end{center}
\end{table}

\begin{figure}[!htb]
\begin{center}
\includegraphics[width=\textwidth, height=2.7in]{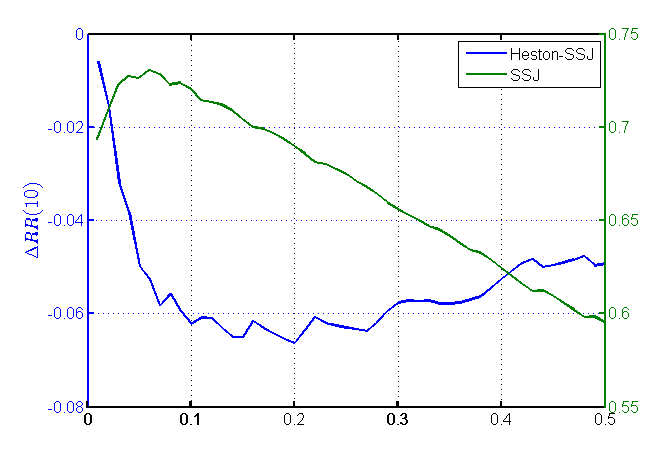}
\caption{RR(10) skew for the Heston and SSJ models for Down-and-out barrier options.}
\label{doSkewJump}
\end{center}
\end{figure}
}

For the diffusion part we again use a non-uniform grid 101 x 81 x 81. The jump non-uniform grid in each direction is a superset of the diffusion grid built using a geometric progression, see \cite{Itkin3D}. So the whole jump grid includes $117 \times 105 \times 103$. The fixed temporal step is 0.01. The typical elapsed time obtained at the same computer is: the first diffusion step - 4.7 secs, the jumps steps in $S,v,R$ - $S$ - 0.15, 0.25, 0.4 secs, the common jumps step - 5.5 secs to converge to the relative accuracy $10^{-2}$, the next jump steps in $S,v,R$ - 0.40, 0.48, 0.58 secs, and the last diffusion step - 4.7 secs. Thus, it takes about 3.5 times more to compute the solution at one step in time as compared with the case with no jumps. This coincides with the theoretical estimation of \cite{ItkinLipton2014}.

The ADI scheme for computing the common and idiosyncratic jumps for the Kou model is described in \cite{ItkinLipton2014}. The value of the ADI parameter $s$ is experimentally chosen to be $s=10000$ which is sufficient to provides convergence of the scheme.

The result of simulation are presented in Fig.~\ref{doSkewJump} where SSJ stays for the stochastic skew model with jumps. It can be seen that even jumps with relatively small intensity as in Table~\ref{TabJumpParam} change the skew by about 10\%. Obviously, the proposed model has the capacity to produce lower of higher skews  by varying the model parameters.

\section{Conclusion}

The paper deals with an advanced model which is build on top of an LSV model by including
the following components: i) the spot/instantaneous variance correlation is stochastic, ii) all stochastic drives - spot, instantaneous variance and their correlation are modeled by \LY processes, iii) all diffusion components are correlated as well as all jump components, however there is no correlation between diffusion and jumps.

This model is used to price FX options and replicate stochastic behavior of the options' skew  which was observed in market data. For the purpose of fast calibration, we consider a forward PIDE and propose a new fully implicit splitting finite-difference scheme to solve it which is fully consistent with the corresponding solution of the backward PIDE. The scheme is unconditionally stable, of second order of approximation in time and space, and achieves a linear complexity in each spatial direction. It also preserves the positivity and norm of the density function which is proven by a series of Propositions. All these results are new.

The results of simulation obtained by using this model demonstrate capacity of the presented approach in modeling stochastic skew. However, it is worth mentioning that despite the proposed model is sufficiently rich, it contains many free parameters. Therefore, calibration of the model, if done at once, could be time-consuming even when using the forward scheme. However, if various market data is available, it is better to calibrate various pieces of the model separately, as this was discussed, e.g., in \cite{Ballotta2014}. Namely, the idiosyncratic jumps first can be calibrated to some marginal distributions using the appropriate instruments. Then the parameters of the common  jumps can be calibrated to the option prices, while keeping parameters of the idiosyncratic jumps fixed. Also the LSV part can be calibrated to the vanilla and exotic option prices as this is usually done, \cite{Bergomi2016}.

\clearpage
\section*{Acknowledgments}
I thank Peter Carr for various useful comments and discussions and Igor Halperin for
some helpful remarks.
I assume full responsibility for any remaining errors.

\section*{References}

\newcommand{\noopsort}[1]{} \newcommand{\printfirst}[2]{#1}
  \newcommand{\singleletter}[1]{#1} \newcommand{\switchargs}[2]{#2#1}

\newpage
\appendix
\section{Splitting scheme for the 3D forward equation}
In this Section we use the HV finite-difference scheme for the three-dimensional case, and derive an explicit representation for the evolutionary backward operator $\mathcal{B}$ in \eqref{BKE1}. Since the HV scheme represents the solution in the form of fractional steps, we need to compress it to a single operator, which will be associated with $\mathcal{B}$.

The HV scheme in the 3D\ case reads, \cite{HoutWelfert2007}
\begin{align} \label{HV3D}
Y_0 &= V_{n-1} + \Delta \tau F(\tau_{n-1},V_{n-1}), \\
Y_j &=  Y_{j-1} + \theta \Delta \tau \left[F_j(\tau_n,Y_j) - F_j(\tau_{n-1},V_{n-1})\right], \ j = 1,2,3, \nn \\
\widetilde{Y}_0 &= Y_0 + \dfrac{1}{2} \Delta \tau \left[F(\tau_n,Y_3) - F(\tau_{n-1},V_{n-1})\right], \nn \\
\widetilde{Y}_j &= \widetilde{Y}_{j-1} + \theta \Delta \tau \left[F_j(\tau_n,\widetilde{Y}_j) - F_j(\tau_n,Y_3)\right], \ j=1,2,3, \nn \\
V_n &= \widetilde{Y}_3, \nn
\end{align}
\noindent where $F = \sum_j F_j, \ j=0,1...k$.

Below for the sake of brevity we denote $F_i^n = F_i(\tau_n)$. We can then write the first equation in \eqref{HV3D} as
\begin{equation} \label{1tr}
Y_0 = (I + \Delta \tau F^{n-1}) V_{n-1},
\end{equation}
where $F(\tau)$ is treated as an operator (or, given a finite-difference grid $\mathbf{G}$, the matrix of the corresponding discrete operator). We also reserve symbol $I$ for an identity operator (matrix).

It is important to notice that operators $F$ at every time step do not explicitly depend on a backward time $\tau$, but only via time-dependence of the model parameters.\footnote{We allow coefficients of the LSV model be time-dependent. However, they are assumed to be piece-wise constant at every time step.}

Proceeding in the same way, the second line of \eqref{HV3D} for $j=1$ is now
\[ (I - \theta\Delta \tau F^n_1) Y_1 = Y_0 - \theta \Delta \tau F^{n-1}_1  V_{n-1}.
\]
Therefore,
\begin{align} \label{2_1_tr}
Y_1 &= (M_1^{n-1})^{-1} \left[Y_0 - \theta \Delta \tau F^{n-1}_1 V_{n-1}\right] \\
&= (M_1^{n-1})^{-1}\Big[I + \Delta \tau \Big(F^{n-1} -\theta F_1^{n-1}\Big)\Big] V_{n-1},
\qquad M_i^n \equiv I - \theta\Delta\tau F^n_i. \nonumber
\end{align}
Similarly, for $j=2$ we have
\begin{equation} \label{2_2_tr}
Y_2 = (M_2^{n-1})^{-1}\left[Y_1 - \theta \Delta \tau F^{n-1}_2 V_{n-1}\right],
\end{equation}
\noindent and for $j=3$
\begin{align} \label{2_3_tr}
Y_3 &= (M_3^{n-1})^{-1}\left[Y_2 - \theta \Delta \tau F^{n-1}_3 V_{n-1}\right] = R_3 V_{n-1}, \\
R_3 &= (M_3^{n-1})^{-1} \Big[ - \dtau \theta F^{n-1}_3 + (M_2^{n-1})^{-1}
\Big[ - \dtau \theta F^{n-1}_2 + (M_1^{n-1})^{-1}
\Big[ I \nonumber \\
&+ \dtau(F^{n-1} - \theta F^n_1) \Big] \Big] \Big]. \nonumber
\end{align}
The third line in \eqref{HV3D} now reads
\begin{equation} \label{3tr}
\widetilde{Y_0} = Y_0 + \frac{1}{2}\Delta \tau \left(F^n Y_3 - F^{n-1} V_{n-1}\right) = \left[I + \frac{1}{2} \Delta \tau \left(F^{n-1} + F^n R_3\right) \right]V_{n-1}.
\end{equation}
The next line in \eqref{HV3D} can be transformed to a chain of expressions
\begin{equation} \label{4_1_tr}
\widetilde{Y_j} = (M_j^n)^{-1} \left(\widetilde{Y}_{j-1} - \theta \Delta \tau F^n_j R_3\right), \quad j=1,2,3.
\end{equation}
Collecting all lines \eqref{2_1_tr}-\eqref{4_1_tr} together we obtain
\begin{align} \label{explR}
V_n &= \mathcal{B} V_{n-1} \\
\mathcal{B} &\equiv (M_3^n)^{-1}
\Big[ - \dtau \theta F_3^n R_3 + (M_2^n)^{-1}
\Big[ - \dtau \theta F_2^n R_3 + (M_1^n)^{-1} \Big[ I
\nonumber \\
&- \dtau \theta F_1^n R_3 + \frac{1}{2} \dtau \left(F^{n-1} + F^n R_3\right)
\Big] \Big] \Big]. \nonumber
\end{align}
Thus, we found an explicit representation for the evolutionary operator $\mathcal(B)$ that follows from the HV finite-difference scheme in \eqref{HV3D}.
To construct the transposed operator $\mathcal{B}^\top = \mathcal{A}$, we use well-known rules of matrix algebra to get
\begin{align} \label{Rtrans}
 \mathcal{A} &= \Big[-\dtau \theta R^\top_3 (F_3^n)^\top +
 \Big[-\dtau \theta R^\top_3 (F_2^n)^\top +
 \Big[ -\dtau \theta R^\top_3 (F_1^n)^\top + I \\
 &+ \frac{1}{2} \dtau \Big((F^{n-1})^\top + R^\top_3 (F^n)^\top\Big)
 \Big](M_1^{n,T})^{-1} \Big] (M_2^{n,T})^{-1} \Big] (M_3^{n,T})^{-1}, \nonumber \\
  R_3^\top &=
  \Big[-\dtau \theta (F_3^{n-1})^\top +
  \Big[-\dtau \theta (F_2^{n-1})^\top +
  \Big[I \nonumber \\
  &+ \dtau \Big((F^{n-1})^\top - \theta (F_1^{n-1})^\top
  \Big)  \Big](M_1^{n-1,T})^{-1} \Big] (M_2^{n-1,T})^{-1} \Big] (M_3^{n-1,T})^{-1}. \nonumber
\end{align}
\noindent where we have assumed $\Delta \tau = \Delta t$.

This scheme can be re-written using a splitting technique, i.e., in the form of the sequential fractional steps, similar to how this is done for the original HV scheme. Omitting an intermediate algebra, we provide just the final result
\begin{align} \label{HV3Dforw}
M_3^{n,\top} Y_3 &= V_{n-1}, \qquad M_2^{n,\top} Y_2 = Y_3, \qquad M_1^{n,\top} Y_1 = Y_2, \\
\widetilde{Y}_0 &= -\dtau \theta \left[ (F_3^{n})^\top Y_3 + (F_2^{n})^\top Y_2
+  \left((F_1^{n})^\top - \dfrac{1}{2\theta} (F^n)^\top \right) Y_1\right],
\nonumber \\
M_3^{n-1,\top} Y_3 &= \widetilde{Y}_0, \nonumber \qquad M_2^{n-1,\top} \widetilde{Y}_2 = \widetilde{Y}_3, \qquad
M_1^{n-1,\top} \widetilde{Y}_1 = \widetilde{Y}_2, \nonumber \\
\widetilde{Y}_4 &= - \dtau \theta
\Big[ (F_3^{n-1})^\top \widetilde{Y}_3 + (F_2^{n-1})^\top \widetilde{Y}_2
+ (F_1^{n-1})^\top \widetilde{Y}_1 \Big]  \nonumber \\
&+ \Big[I + \dtau (F^{n-1})^\top\Big]\widetilde{Y}_1, \nonumber \\
V_n &= \Big[I + \dfrac{1}{2} \dtau (F^{n-1})^\top \Big] Y_1
+    \widetilde{Y}_4 . \nonumber
\end{align}
As mentioned in \cite{Itkin2014b}, this scheme, however, has two problems. First, when using splitting (or fractional steps), one usually wants all internal vectors $Y_j$, $j \in [0,3]$ and $\widetilde{Y}_k$, $k \in [0,3]$ to form consistent approximations to $V_n$. The scheme in \eqref{HV3Dforw} loses this property at step 4. Second, because at step 4 the norm of the matrix on the right-hand side is small, the solution is sensitive to round-off errors.

This issue can be resolved by using the method ("trick")\ described in \cite{Itkin2014b}.
We add and subtract $R_3^\top V_{n-1}$ to the first line of \eqref{Rtrans}, and represent the second $R_3^\top V_{n-1}$ as
\[ R_3^\top V_{n-1} = Y_1 - \dtau \theta
\left[ (F_3^{n-1})^\top Y_3 + (F_2^{n-1})^\top Y_2
+ \left((F_1^{n-1})^\top - \frac{1}{\theta}(F^{n-1})^\top\right) Y_1 \right]. \]
Then the final scheme reads
\begin{align} \label{HV3DforwTr}
M_3^{n,\top} Y_3 &= V_{n-1}, \qquad M_2^{n,\top} Y_2 = Y_3, \qquad M_1^{n,\top} Y_1 = Y_2, \\
\widetilde{Y}_0 &= V_{n-1} -\dtau \theta \left[ \sum_{i=1}^3 (F_i^{n})^\top Y_i  - \dfrac{1}{2\theta} (F^n)^\top Y_1\right],\nonumber \\
M_3^{n-1,\top} \widetilde{Y}_3 &= \widetilde{Y}_0, \qquad M_2^{n-1,\top} \widetilde{Y}_2 = \widetilde{Y}_3, \qquad
M_1^{n-1,\top} \widetilde{Y}_1 = \widetilde{Y}_2, \nonumber \\
Z_1 &= \left(1 + \frac{1}{2} \dtau F^{n-1,\top}\right) Y_1, \qquad Z_2 = \left(1 + \dtau F^{n-1,\top}\right) \Delta Y_1, \nonumber \\
p_n &= Z_1 - \dtau \theta \sum_{i=1}^3 (F_i^{n-1})^\top \Delta Y_i + Z_2. \nonumber
\end{align}
\noindent where $\Delta Y = \widetilde{Y} - Y$. The final step is to put $\dtau = \Delta t$.

\end{document}